\documentclass{llncs}
\usepackage{amsmath,amssymb,amstext,amsfonts,graphics, enumerate,boxedminipage}
\usepackage{color,graphicx}
\def\Z{{\mathbb Z}}
\def\R{{\mathbb R}}
\newcommand{\NP}{{\sf NP}}

\begin{document}
\title{The Stable Fixtures Problem with Payments\thanks{An extended abstract of this paper has appeared in the proceedings of WG 2015~\cite{BKPW15}.}}

\author{P\'eter Bir\'o\inst{1,2}
\thanks{Supported by the Hungarian Academy of Sciences under  Momentum Programme
LD-004/2010, by the Hungarian Scientific Research Fund - OTKA (no.\
K108673), and by J\'anos Bolyai Research Scholarship of the Hungarian
Academy of Sciences.}
\and Walter Kern\inst{3} \and Dani\"el Paulusma\inst{4}
\thanks{Supported by EPSRC Grant EP/K025090/1.}
 \and P\'eter Wojuteczky\inst{1}}

\institute{
Institute of Economics, Hungarian Academy of Sciences; H-1112, Buda\"orsi \'ut 45, Budapest, Hungary, \texttt{peter.biro@krtk.mta.hu, peter.wojuteczky@gmail.com}
\and
Department of Operations Research and Actuarial Sciences, Corvinus University of Budapest
\and
Faculty of Electrical Engineering, Mathematics and Computer Science,\\
University of Twente, P.O.Box 217, NL-7500 AE Enschede\\
\texttt{w.kern@math.utwente.nl}
\and
School of Engineering and  Computing Sciences, Durham University,\\
Science Laboratories, South Road,
Durham DH1 3LE, UK,\\
\texttt{daniel.paulusma@durham.ac.uk}
}

\maketitle

\begin{abstract}
We generalize two well-known game-theoretic models by
introducing
multiple partners matching games,
defined by a
graph $G=(N,E)$, with an integer vertex capacity function~$b$ and an edge weighting $w$. The set $N$ consists of a number of players that are to form a set $M\subseteq E$ of 2-player coalitions $ij$ with value
$w(ij)$, such that each player $i$ is in at most $b(i)$ coalitions.
A payoff vector
 is a mapping $p: N \times N \rightarrow \R$ with
  $p(i,j)+p(j,i)=w(ij)$ if $ij\in M$ and $p(i,j)=p(j,i)=0$ if $ij\notin M$.
The pair $(M,p)$ is called a solution.
A pair of players $i,j$ with $ij\in E\setminus M$ blocks a solution $(M,p)$ if
$i, j$ can form,  possibly only after withdrawing from one of their existing 2-player coalitions,
a new 2-player coalition in which they are mutually better off.
A solution is stable if it has no blocking pairs.
Our contribution is as follows:
\begin{itemize}
\item We survey, for the first time, known results on stable solutions in seven basic models and show that our model of multiple partners matching games is the natural model that was missing so far.
\item We give a polynomial-time algorithm that
either finds that a given multiple partners matching game has
no stable solution, or obtains a stable solution for it.
Previously this result was only known for multiple partners assignment games, which correspond to the case where $G$ is bipartite (Sotomayor, 1992) and for
the case where $b\equiv 1$ (Bir\'o et al., 2012).
\item We characterize the set of stable solutions of a multiple partners matching game
 in two different ways and show how this leads to simple proofs for a number of known results of Sotomayor (1992,1999,2007) for multiple partners assignment games
 and to generalizations of some of these results to multiple partners matching games.
 \item We perform a study on the core of the corresponding cooperative game, where coalitions of any size may be formed.
 In particular we show that the standard relation between the existence of a stable solution and the non-emptiness of the core, which holds in the other
 models with payments, is no longer valid for our (most general) model.
 We also prove that the problem of deciding if an allocation belongs to the core jumps from being polynomial-time solvable for $b\leq 2$ to \NP-complete for $b\equiv 3$.
\end{itemize}

\end{abstract}

\setcounter{footnote}{0}
\section{Introduction}\label{s-intro}

Consider a group of soccer teams participating in a series of friendly games with each other
off-season.
Suppose each team has some specific target number of games it wants to play. For logistic reasons, not every two teams
can play against each other. Each game brings in some revenue, which
is to be shared by the two teams involved.
The revenue of a game may depend on several factors, such as the popularity of the two teams involved or
 the soccer stadium in which the game is played. In particular, at the time when the schedule for these games is prepared, the expected gain may well depend on future outcomes in the current season (which are in general  difficult to predict~\cite{KP01}). In this paper, we assume for simplicity that the revenues are known.
 Is it possible to construct a {\it stable} fixture of games, that is, a schedule such that there exist no two unmatched teams that are better off by playing against each other?
Note that if teams decide to play against each other, they may first need to
cancel one of their other games in order not to exceed their targets.

The above example describes the problem introduced in this paper
(see Section~\ref{s-alternative} for another example).
In the next section we explain how we model this problem.

\subsection{Our Model}

We model the above example in two settings, namely as a matching problem and as a cooperative game. As we will show these two settings are deeply interwoven.

\medskip
\noindent
{\bf Matching Problem.}
A {\it multiple partners matching game} is a
 triple $(G,b,w)$, where $G=(N,E)$ is a finite undirected graph on $n$ vertices and $m$ edges  with no loops and no multiple edges,
$b:N\to \Z_+$ is a  {\it vertex capacity function}, which is a nonnegative integer function, and
$w:E\to \R_+$ is a nonnegative edge weighting.
The set~$N$ is called the {\it player set}. There exists an edge $ij\in E$ if and only if players~$i,j$ can form a 2-player coalition.
A set $M\subseteq E$ is a {\it $b$-matching} if every player $i$ is incident to at most $b(i)$ edges of $M$. So, a $b$-matching is
a set of 2-player coalitions, in which no player is involved in more 2-player coalitions than
described by her capacity.
If $ij\in M$  then  $i$ and $j$ are
{\it matched} by $M$;
we also say that $i$ and $j$ are {\it partners} under $M$.
The {\it value} of a 2-player coalition ${i,j}$ with $ij\in E$ is given by~$w(ij)$.

A nonnegative function
$p: N\times N\to \R_+$ is a
{\it payoff vector} with respect to a $b$-matching~$M$ if
the following two conditions hold:
\begin{itemize}
\item [$\bullet$]  $p(i,j)+p(j,i)=w(ij)$ for all $ij\in M$;
\item [$\bullet$]  $p(i,j)=p(j,i)=0\;\;\;\;\;\;$ for all $ij\notin M$.
\end{itemize}
In that case we also say that
$M$ and~$p$ are
{\it compatible}.\footnote{Assume that $b$ and $w$ are strictly positive functions. If $p$ and $M$ are compatible, then $p$ is not compatible with
any other $b$-matching, that is, $p$ uniquely determines $M$. However, for our purposes it is more convenient to follow the literature and define $p$ with respect to $M$.}
Note that $p$ prescribes how the value $w(ij)$ of a 2-player coalition $\{i,j\}$ is  distributed amongst $i$ and~$j$, ensuring that non-coalitions between two players yield a zero payoff.
A pair $(M,p)$, where $M$ is a $b$-matching and $p$ is a payoff compatible with $M$, is  a {\it solution} for $(G,b,w)$.
We view $p$ as a vector with entries $p(i,j)$,
which we call {\it payoffs}.

Let $(M,p)$ be a solution.
Two players $i,j$ with $ij\in E\setminus M$ may decide to form a new 2-player coalition if they are ``better off'',
even if one or both of them must first leave an existing 2-player coalition in $M$ (in order not to exceed their individual capacity).
To describe this formally we  define a {\it utility function} $u_p: N\to \R_+$, related to a payoff vector~$p$.
If $i$ is {\it saturated} by $M$, that is, if $i$ is incident with $b(i)$ edges in $M$, then we let
$u_p(i)= \min\{p(i,j): ij\in M\}$ be the worst payoff $p(i,j)$ of any 2-player coalition~$i$ is involved in. Otherwise, $i$ is {\it unsaturated} by $M$ and we define $u_p(i)=0$.
Alternatively, we could define $u_p(i)$ as the $b(i)$th largest payoff $p(i,j)$ to $i$. Note that the second definition shows that utilities are independent of $M$ and determined by $p$ only (recall$^1$ that this is because $p$ in fact determines $M$).

A pair~$i,j$ with $ij\in E\setminus M$ {\it blocks} $(M,p)$ if  $u_p(i)+u_p(j)<w(ij)$.
We say that $(M,p)$ is \emph{stable} if it has no blocking pairs, or equivalently, if
every edge $ij\in E\setminus M$ satisfies the {\it stability condition}, that is, if
\[
u_p(i)+u_p(j)\geq w(ij)\; \mbox{for all}\; ij\in E\setminus M.
\]
Note that the stability condition only needs to be verified for edges not in $M$.

\begin{remark}\label{r-bis1}
Let $(G,b,w)$ be a multiple partners matching game with $b\equiv 1$.
Then any $b$-matching is a 1-matching, i.e, a {\it matching}, as for each $i\in N$, we have $p(i,j)>0$ for at most one player $j\neq i$, which must be matched to $i$.
In that case we will sometimes
assume, with slight abuse of notation, that $p$ is a nonnegative function defined on~$N$.
Then we can write $u_p(i)=p(i)$ for every $i\in N$.
Checking whether a pair  $(M,p)$ is a solution for $(G,1,w)$ comes down to verifying whether $p(i)+p(j)=w(ij)$ holds for every edge $ij\in M$.
Checking whether a solution $(M,p)$ is stable comes down to verifying whether $p(i)+p(j)\geq w(ij)$ holds for every edge $ij\in E\setminus M$.
\end{remark}

We can now define our problem formally:

\begin{center}
\begin{boxedminipage}{.99\textwidth}
\textsc{Stable Fixture with Payments}(SFP)\\[2pt]
\begin{tabular}{ r p{0.8\textwidth}}
\textit{~~~~Instance:} &a multiple partners matching game $(G,b,w)$\\
\textit{Question:} &does $(G,b,w)$ have a stable solution?
\end{tabular}
\end{boxedminipage}
\end{center}

\noindent
{\it Example 1.}
Let $G$ be the 4-vertex cycle $u_1v_1u_2v_1u_1$.
Let  $b\equiv 1$ and $w\equiv 1$. Then $G$ has two maximum weight matchings, namely $M=\{u_1v_1,u_2v_2\}$ and $\hat{M}=\{u_1v_2,u_2v_1\}$. Let $p$ be given by $p(u_1,v_1)=\frac{7}{10}$, $p(v_1,u_1)=\frac{3}{10}$, $p(u_2,v_2)=\frac{7}{10}$, $p(v_2,u_2)=\frac{3}{10}$ and $p(u_1,v_2)=p(v_2,u_1)=p(u_2,v_1)=p(v_1,u_2)=0$. Then $p$ is compatible with $M$ and $(M,p)$ is a stable solution for $(G,1,1)$. We also observe that $p$ is not compatible with $\hat{M}$.
However, there exists a stable solution $(\hat{M},\hat{p})$, where $\hat{p}$ can be obtained from $p$ by permuting the entries $p(i,j)$ for every fixed $i$.
Namely, let $\hat{p}$ be defined as
$\hat{p}(u_1,v_2)=\frac{7}{10}$, $\hat{p}(v_2,u_1)=\frac{3}{10}$, $\hat{p}(u_2,v_1)=\frac{7}{10}$, $\hat{p}(v_1,u_2)=\frac{3}{10}$ and $\hat{p}(u_1,v_1)=\hat{p}(v_1,u_1)=\hat{p}(u_2,v_2)=\hat{p}(v_2,u_2)=0$.

\medskip
\noindent
Being a bit more specific than in the above example, let $(M,p)$ and $(\hat{M},\hat{p})$ be two stable solutions for  a multiple partners matching game $(G,b,w)$. Then we say that $p$ and $\hat{p}$ are {\it equivalent} if the following four conditions hold:
\begin{itemize}
\item $u_p(i)=u_{\hat{p}}(i)$ for every $i\in N$,
\item $p(i,j)=\hat{p}(i,j)$ and $p(j,i)=\hat{p}(j,i)$  for every $ij\in M \cap \hat{M}$,
\item $p(i,j)=u_p(i)=u_{\hat{p}}(i)$ and $p(j,i)=u_p(j)=u_{\hat{p}}(j)$ for every $ij\in M\setminus\hat{M}$, and
 \item $\hat{p}(i,j)=u_{\hat{p}}(i)=u_p(i)$ and $\hat{p}(j,i)=u_{\hat{p}}(j)=u_p(j)$ for every $ij\in \hat{M}\setminus M$.
\end{itemize}
If  $p$ and $\hat{p}$ are equivalent, then
the multisets $\{p(i,j)\; |\; j\in N\; \mbox{with}\; ij\in E\}$ and $\{\hat{p}(i,j)\; |\; j\in N\; \mbox{with}\; ij\in E\}$ are the same for every $i\in N$.

For a payoff vector~$p$, the {\it total payoff vector} $p^t\in \R^N$ is defined by
$p^t(i)=\sum_{j:ij\in E}p(i,j)$ for every $i\in N$.
Note that if two payoff vectors $p$ and $\hat{p}$ are equivalent, then their total payoff vectors are the same, that is, $p^t=\hat{p}^t$.\footnote{In fact, Sotomayor\cite{Sotomayor92}, calls two payoff vectors $p$ and $\hat{p}$ equivalent if $p^t=\hat{p}^t$ and the
multisets $\{p(i,j)\; |\; j\in N\; \mbox{with}\; ij\in E\}$ and $\{\hat{p}(i,j)\; |\; j\in N\; \mbox{with}\; ij\in E\}$ are the same for every $i\in N$.
Our definition implies hers and enables us to prove our results in a stronger form. However, in the proof of her results, Sotomayer uses in fact our definition as well.}
Moreover, for $b\equiv 1$ two payoff vectors $p$ and $\hat{p}$ are equivalent if and only if
$p^t=\hat{p}^t$.
We illustrate the definition of two payoff vectors being equivalent with another example.

\bigskip
\noindent
{\it Example 2.}
Let $G$ be the bipartite graph displayed in Figure~\ref{f-exa} that  has six vertices $u_1$, $u_2$, $u_3$, $v_1$, $v_2$, $v_3$ and edges $u_1v_1$, $u_1v_2$, $u_1v_3$, $u_2v_1$, $u_2v_2$, $u_3v_1$, $u_3v_2$ and $u_3v_3$ (that is, all edges between the two partition classes exist except $u_2v_3$). Let $b(u_1)=b(u_2)=b(v_1)=2$ and $b(u_3)=b(u_3)=b(v_3)=1$, and define $w(u_1v_1)=4$, $w(u_1v_2)=6$, $w(u_1v_3)=5$, $w(u_2v_1)=4$, $w(u_2v_2)=1$, $w(u_3v_1)=1$, $w(u_3v_2)=3$ and $w(u_3v_3)=2$. Then $G$ has two maximum weight $b$-matchings, namely $M=\{u_1v_1,u_1v_2,u_2v_1,u_3v_3\}$ and $\hat{M}=\{u_1v_1,u_1v_3,u_2v_1,u_3v_2\}$. Let $p$ be given by $p(u_1,v_1)=3$, $p(v_1,u_1)=1$, $p(u_1,v_2)=3$, $p(v_2,u_1)=3$, $p(u_2,v_1)=2$, $p(v_1,u_2)=2$, $p(u_3,v_3)=0$, $p(v_3,u_3)=2$ and $p(u_1,v_3)=p(v_3,u_1)=p(u_2,v_2)=p(v_2,u_2)=p(u_3,v_1)=p(v_1,u_3)=p(u_3,v_2)=p(v_2,u_3)=0$. Then $p$ is compatible with $M$ and $(M,p)$ is a stable solution for $(G,b,w)$. We also observe that $p$ is not compatible with $\hat{M}$.
However, there exists a payoff vector~$\hat{p}$ that is equivalent with $p$ such that  $(\hat{M},\hat{p})$ is a stable solution.
Namely, let $\hat{p}$ be defined as
$\hat{p}(u_1,v_1)=3$, $\hat{p}(v_1,u_1)=1$, $\hat{p}(u_1,v_3)=3$, $\hat{p}(v_3,u_1)=2$, $\hat{p}(u_2,v_1)=2$, $\hat{p}(v_1,u_2)=2$, $\hat{p}(u_3,v_2)=0$, $\hat{p}(v_2,u_3)=3$ and $\hat{p}(u_1,v_2)=\hat{p}(v_2,u_1)=\hat{p}(u_2,v_2)=\hat{p}(v_2,u_2)=\hat{p}(u_3,v_1)=\hat{p}(v_1,u_3)=\hat{p}(u_3,v_3)=\hat{p}(v_3,u_3)=0$.

\begin{figure}
\begin{center}
\scalebox{0.6}{\includegraphics{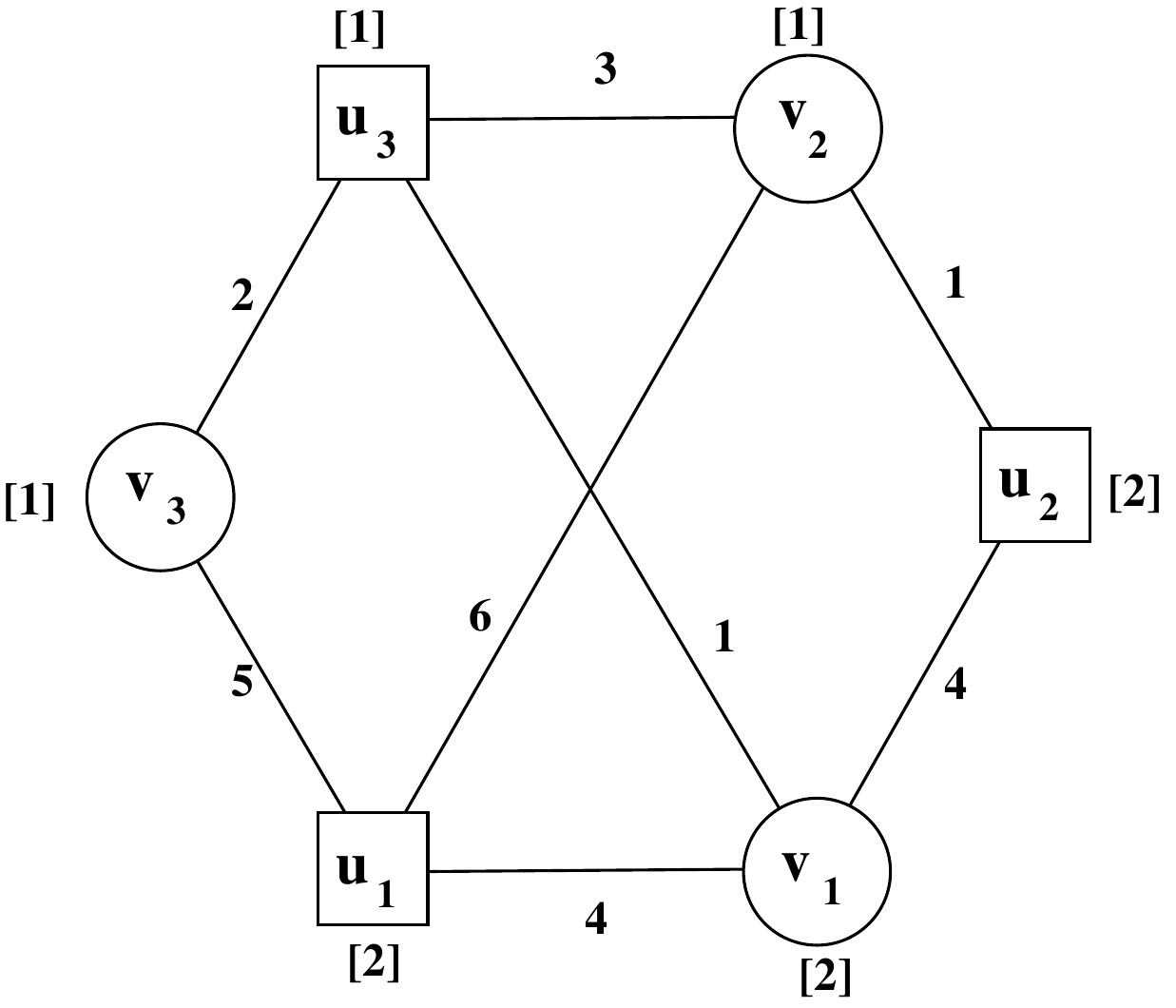}}
\caption{The bipartite graph $G$ from Example 2 with the vertex capacities and edge weights.}\label{f-exa}
\end{center}
\end{figure}

\noindent
{\bf Cooperative Game.}
So far, we modelled only situations in which 2-player coalitions can be formed. Allowing coalitions of any size
is a natural and well-studied setting in the area of Cooperative Game Theory. Moreover, as we will discuss,
there exist close relationships between stable solutions and their counterpart in the second setting, the so-called core allocations, which we define below.

A \emph{cooperative game  with transferable utilities} (TU-game) is a pair $(N,v)$ consisting of a set $N$ of $n$ \emph{players} and a \emph{value function} $v: 2^N\to \R_+$ with $v(\emptyset) = 0$.
It is usually assumed that  the {\it grand coalition} $N$ is formed. Then the central problem is how to allocate the
{\it total value} $v(N)$ to the individual players in $N$.
In this context, a \emph{payoff vector} (or {\it allocation}) is a vector $p \in \R^N$ with $p(N) = v(N)$, where we write $p(S) = \sum_{i\in S}p(i)$ for $S\subseteq N$.
The {\it core} of a TU-game consists of all allocations $p \in \R^N$ satisfying
\begin{equation} \label{coreq}
\begin{array}{rrrl}
  p(S)&\ge&v(S),&\;\;\;\emptyset \neq S \subseteq N\\
  p(N)&=&v(N).
 \end{array}
\end{equation}
A core allocation is seen as reasonable, because it offers no incentive for a subset of players to leave the grand coalition and form a coalition on their own.
However,
a TU-game may have an empty core. Hence, the most interesting computational complexity
problems (given an input game) are:
\begin{itemize}
\item [{\bf P1.}] Determine if the core is non-empty;
\item [{\bf P2.}] Exhibit a vector in the core (provided there is any);
\item [{\bf P3.}] Determine if a given $p \in \R^N$  is in the core or find a coalition $S$ with $p(S)<v(S)$.
\end{itemize}
In the literature both polynomial-time and
(co-)\NP-hardness results are known for each of these three problems
(see e.g.~\cite{DIN99}).
An efficient algorithm for answering P3 implies that P1 and P2 can be solved in polynomial time as well.
This follows from the work of
 Gr\"otschel, Lov\'asz and Schrijver~\cite{GLS81,GLS93} who proved, by refining the ellipsoid method of
 Khachiyan~\cite{Kh79}, that an efficient algorithm for solving the separation problem for a polyhedron~$P$ implies a polynomial-time algorithm that either finds that $P$ is empty, or  obtains a vector of~$P$.

We define the TU-game~$(N,v)$ that corresponds with a multiple partners matching game $(G,b,w)$ by setting, for every $S\subseteq N$,
$$v(S)=w(M_S)=\sum_{e\in M_S}w(e),$$
where $M_S$ is a maximum weight $b$-matching
in the graph $G[S]$, that is, the subgraph of $G$ induced by $S$.
We define $v(S)=0$ if $S$ induces an edgeless graph.
We say that $(N,v)$ is {\it defined on} $(G,b,w)$ but, unless confusion is possible, we may also call $(N,v)$ a multiple partners matching game.
If we say that a payoff vector~$p$ involved in a stable solution $(M,p)$ of SFP is a core allocation, we mean in fact that the total payoff vector~$p^t$ is a core allocation.

\medskip
\noindent
{\it Example 3.} Let $G$ be the 4-vertex cycle $v_1v_2v_3v_4v_1$.  Define a vertex capacity function~$b$ by $b(v_1)=b(v_2)=1$ and
$b(v_3)=b(v_4)=2$, and an edge weighting $w$ by $w(v_1v_2)=3$ and $w(v_2v_3)=w(v_3v_4)=w(v_4v_1)=1$.
The pair $(M,p)$ with $M=\{v_1v_2,v_3v_4\}$ and $p(v_1,v_2)=p(v_2,v_1)=\frac{3}{2}$, $p(v_3,v_4)=p(v_4,v_3)=\frac{1}{2}$ and
$p(v_2,v_3)=p(v_3,v_2)=p(v_4,v_1)=p(v_1,v_4)=0$ is a solution for the multiple partners matching game $(G,b,w)$. Note that $u_p(v_1)=u_p(v_2)=\frac{3}{2}$ and $u_p(v_3)=u_p(v_4)=0$. We find that $(M,p)$ is even a stable solution, because
$u_p(v_2)+u_p(v_3)=\frac{3}{2}\geq 1=w(v_2v_3)$ and $u_p(v_4)+u_p(v_1)=\frac{3}{2}\geq 1=w(v_4v_1)$ (note that we only need to verify the stability condition for edges outside the matching).
Moreover, the total payoff vector $p^t$ given by $p^t(v_1)=p^t(v_2)=\frac{3}{2}$ and $p^t(v_3)=p^t(v_4)=\frac{1}{2}$ is readily seen to be a core allocation of the corresponding TU-game. In Section~\ref{s-core} we will give an example of a multiple partners matching game with no stable solutions for which the corresponding TU-game has a non-empty core.

\medskip
\noindent
Before stating our results for multiple partners matching games in both settings we first discuss some existing work. As we will see, our model in both settings
generalizes (or relaxes) several well-known models.

\subsection{Known Results}\label{s-known}
The first model that we discuss is the most basic model related to
the famous {\it stable marriage problem}
(SM),
defined
as follows.
Given two disjoint sets $I$ and~$J$ of men and women, respectively, let each player have a strict preference ordering over
 a subset of players from the opposite set.
  A set of marriages is a matching in the underlying
bipartite graph with partition classes $I$ and $J$. Such a matching is stable if there is no unmarried pair, who would prefer to marry each other instead of  a possible other partner. Gale and Shapley~\cite{GS62} proved that every instance of this problem has a stable matching and gave a
linear-time algorithm that finds
 one.\footnote{Gale and Shapley assumed that the underlying bipartite graph is complete, but in fact their algorithm can also be used for arbitrary graphs; see~\cite{GI89}.}

The main assumptions in the basic model are
\begin{itemize}
\item [(i)] monogamy: each player is matched to at most one other player (1-matching);
\item [(ii)] opposite-sex: every match is between players from $I$ and $J$ (bipartiteness);
\item [(iii)] no dowry: only cardinal preferences are considered (no payments).
\end{itemize}
Dropping one or more of these three conditions leads to seven
other
models, one of which corresponds to our
new model
of multiple partner matching games, namely the one, in which {\it none} of the three conditions (i)--(iii) is imposed. Below we briefly survey the other six models (see also Table~\ref{heuristics}).

In the first three models that we discuss, payments are not allowed. Hence, the notion of a (core) allocation is meaningless for these three models.

\medskip
\noindent
{\it Remove (i).}
If we allow bigamy, that is, if we allow general $b$-matchings instead of only 1-matchings, we obtain the
{\it many-to-many stable matching problem}, which
 generalizes the stable marriage problem.
The problem variant, in which we demand that $b(i)=1$ for each player~$i\in I$, is called the
{\it college admission problem}~\cite{GS62}, which is  also known as the {\it many-to-one stable matching problem}~\cite{RS90} and as the {\it hospital/residents problem}~\cite{Ma13}.
Gale and Shapley~\cite{GS62} proved that every instance of the college admission problem has a stable matching and gave a linear-time algorithm that finds one.
Ba\"iou and Balinski~\cite{BB00} proved these two results for the (more general)
many-to-many stable matching problem.

\medskip
\noindent
{\it Remove (ii).}
If we allow same-sex marriages, so
the underlying graph may be an arbitrary graph that is not necessarily  bipartite,
then we get the \emph{stable roommates problem}
(SR),
also defined by Gale and Shapley \cite{GS62}. They proved that, unlike the previously discussed models, in this model a stable matching does not always
exist. Irving~\cite{Irving85} gave a linear-time algorithm that finds a stable matching if there exists
one.\footnote{Irving showed that a stable matching may not exist even if the underlying graph is complete. His algorithm was designed for complete graphs, but can be used for arbitrary graphs~\cite{GI89}.}

\medskip
\noindent
{\it Remove (i) \& (ii).}
Allowing bigamy {\it and} same-sex marriages leads to
the \emph{stable fixtures problem}
(SF),
which
generalizes the stable roommates problem. Hence, a stable matching does not always exist.
Irving and Scott~\cite{IS07}
gave a linear-time algorithm for finding a stable matching (if there exists one).
Cechl\'arov\'a and Fleiner~\cite{CF05} defined the more general \emph{multiple activities problem}, in
which the underlying graph may have multiple edges. They proved that even in this setting
a stable matching can be found in polynomial time (if there exists one).
Moreover, they also showed that this problem can be reduced to
SR by a polynomial size graph construction.

\medskip
\noindent
In the remaining three models we allow payments to individual players.

\medskip
\noindent
{\it Remove (iii).} If we allow dowry then we
obtain an {\it assignment game}, which is a multiple partners matching game $(G,b,w)$ where $G$ is bipartite and $b\equiv 1$.
In this case the SFP problem is known as the {\it stable marriage problem with payments} (SMP).
Koopmans and Beckmann~\cite{KB57e} proved the following result.

\begin{theorem}[\cite{KB57e}]\label{t-kb}
Every instance $(G,1,w)$ of  {\sc SMP} has a stable solution, which can be found in polynomial time. If $(M,p)$ is a stable solution then $M$ is a maximum weight matching.
Moreover, for every other maximum weight matching $\hat{M}$ of $G$ there exists an equivalent payoff vector~$\hat{p}$ of $p$ that is compatible with $\hat{M}$.
\end{theorem}

Shapley and Shubik~\cite{SS72} proved that every core allocation of an assignment game is a payoff vector in a stable solution for the corresponding instance of {\sc SMP} and vice versa.
Combining their result with Theorem~\ref{t-kb} implies that every assignment game has a non-empty core.
It is possible to obtain a stable solution in polynomial time and also to
give affirmative answers to problems P1-P3 about the core of an assignment game; in the next paragraph we explain that this holds even if we allow same-sex marriages.

\medskip
\noindent
{\it Remove (ii) \& (iii).} If  we allow dowry and same-sex marriages then we obtain a {\it matching game}, which is
 a multiple partners matching game $(G,b,w)$ where $b\equiv 1$.  In this case the SFP problem is called the {\it stable roommates problem with payments} (SRP).
The following two observations are well-known~\cite{BKP12,EK01} and easy to verify.

First, a payoff vector~$p$ is a core allocation of a matching game if and only if there exists a matching $M$ such that $(M,p)$ is a stable solution.
Note that the core may be empty (take a triangle with unit edge weights). Hence, a stable solution may not always exist.

Second, for matching games, the coalitions in the system of inequalities~(\ref{coreq}) may be restricted to 2-player coalitions.
This
means that problem~P3, on core membership, can be answered in linear time.
As explained, a polynomial-time algorithm for P3 also leads to
polynomial-time algorithms for solving P1 and P2, about core non-emptiness, and finding a core allocation, and thus finding a stable solution. The restriction to 2-player coalitions even allows one to
use the ellipsoid method of  Khachiyan~\cite{Kh79} directly.
In a previous paper~\cite{BKP12}, we circumvented
the ellipsoid method and presented an
$O(nm+n^2\log n)$-time algorithm that either finds that the core is empty, or obtains a core allocation.

\medskip
\noindent
{\it Remove (i) \& (iii).} If  we allow dowry and  bigamy  then we obtain a {\it multiple partners assignment game}, which is  a multiple partners matching game $(G,b,w)$ where $G$ is bipartite.
In this case the SFP problem is called the {\it multiple partners assignment problem} (MPA).
Just as matching games, multiple partners assignment games generalize assignment games.
Sotomayor generalized Theorem~\ref{t-kb}.

\begin{theorem}[\cite{Sotomayor92}]\label{sot92}
Every instance $(G,b,w)$ of {\sc MPA} has a stable solution, which can be found in polynomial time.
If $(M,p)$ is a stable solution then $M$ is a maximum weight $b$-matching and $p^t$ is a core allocation of the corresponding
multiple partners
assignment game.
Moreover, for every other maximum weight $b$-matching $\hat{M}$ of $G$ there exists an equivalent payoff vector~$\hat{p}$ of $p$ that is compatible with $\hat{M}$.
\end{theorem}
Theorem~\ref{sot92} gives affirmative answers for problems~P1 and~P2 for multiple partners assignment games. The answer to~P3 was still open for multiple partners assignment games and is settled in this paper, as we will discuss in Section~\ref{s-our}.
We refer to Table~\ref{heuristics} for a survey of known results and related new results that we show in this paper. In Section~\ref{s-our} we describe all our new results.

\begin{table}
\begin{center}
\small
\begin{tabular}{|l|l|l|l|}
\cline{3-4}
 \multicolumn{2}{c|}{} & $\;$opposite-sex   & $\;$same-sex  allowed\\
  \multicolumn{2}{c|}{}&$\;$(bipartite graphs)  & $\;$(general graphs)\\
 \hline
 monogamy  & no dowry & $\;$SM (stable marriage)  & $\;$SR (stable roommates) \\
(1-matching)& &{\it always a stable solution}
& {\it not always a stable solution}\\
\cline{2-4}
 & $\;\;\;\;$dowry & $\;$SMP / assignment game & $\;$SRP / matching game\\
 &(payments)& {\it always a stable solution}
 &{\it not always a stable solution}\\
 && {\it core always non-empty}
 &{\it core may be empty}\\
\hline
 bigamy & no dowry$\;$ & $\;$many-to-many stable matching  & $\;$SF (stable fixtures)\\
 allowed && {\it always a stable solution}
 &{\it not always a stable solution}\\
\cline{2-4}
  & $\;\;\;\;$dowry & $\;$MPA / mp assignment game &
 \textbf{$\;$SFP / mp matching game\;}\\
&& {\it always a stable solution}
&{\it not always a stable solution}\\
&& {\it core always non-empty}
&{\it core may be empty}\\
\hline
\end{tabular}\\[12pt]
\caption{The eight different models; ``mp'' stands for ``multiple partners''. The {\bf highlighted} case is
the new model introduced and considered in this paper. As discussed in Section~\ref{s-known}, for all models, the results on the (non-)existence of a stable solution or core allocation
were known already or can be readily deduced from known results.
Moreover, as discussed in Section~\ref{s-known} as well, it was already known that
for all models, except for the highlighted model, a stable solution can be found in polynomial time (if it exists) and that for all the models
with payments, except for the highlighted model, a core allocation can be computed in polynomial time (if the core is non-empty). As explained in
Section~\ref{s-our},
we give a polynomial-time algorithm for finding a stable solution (if it exists) for the highlighted model. We also prove that every stable solution corresponds to a core allocation.
However, in contrast to the correspondence between the existence of stable solutions for SRP and core allocations for matching games,
we show that there exist multiple partners matching games with a non-empty core, while the corresponding instance of SFP has no stable solution.}
\label{heuristics}
\normalsize
\end{center}
\end{table}

\subsection{Our Results}\label{s-our}

In Section~\ref{s-fix}, after stating some known results that we need as lemmas in Section~\ref{s-prelim}, we will prove that SFP is polynomial-time solvable. This generalizes the aforementioned corresponding results for
SRP and MPA, respectively. In particular we prove that also for multiple partners matching games, the payoff vectors in stable solutions are always
core allocations
and the compatibility result of Theorem~\ref{sot92} for maximum weight $b$-matchings of bipartite graphs can be generalized to
arbitrary graphs
(in fact we generalize the whole of Theorem~\ref{sot92} from MPA to SFP apart from the opening claim that every instance of MPA has a stable solution, which does not hold for SFP).
Our technique for proving all these results is based on a reduction to MPA. Moreover, we  characterize the set of stable solutions for a given instance of SFP via a reduction to SRP. We do this via linear programming techniques that show a close relationship between optimal solutions in the dual LP for SFP and stable solutions in the reduced instance of SRP.

In Section~\ref{s-alternative} we illustrate the power of our characterization of the set of stable solutions for an instance of SFP by proving that this characterization, which obviously holds for instances of MPA as well,
immediately leads to simple alternative proofs of the three main results of Sotomayor for MPA,
namely Theorem~\ref{sot92} of~\cite{Sotomayor92} (the whole theorem) and two theorems of~\cite{Sotomayor99} and~\cite{Sotomayor07}, respectively.

In Section~\ref{s-core} we aim to increase our understanding of the core of a multiple partners matching game.
To that end we first prove that every core allocation of a multiple partners matchings game equals the total payoff vector of some suitable payoff vector, which we can find in polynomial time.
We then show that
there exist multiple partners matching games with a non-empty core for which the corresponding instance of {\sc SFP} has no stable solutions.
Afterwards we focus on
core membership, which corresponds to problem~P3.
We first show that core membership is polynomial-time solvable for multiple partner matching games defined on a
triple $(G,b,w)$ with $b\leq 2$, that is, with $b(i)\leq 2$ for all $i\in N$.
Due to the aforementioned result of  Gr\"otschel, Lov\'asz and Schrijver~\cite{GLS81,GLS93} this leads to efficient answers to P1 and P2 as well (for $b\leq 2$). In our proof, we make a connection to the tramp steamer problem~\cite{L76}. We complement this result by showing
that it is tight, namely that already
for $b\equiv 3$ and $w\equiv 1$, testing core membership is co-\NP-complete even for multiple partners
{\it assignment} games (whose core is always non-empty and for which a core allocation can be found in polynomial time due to Theorem~\ref{sot92}).

Finally, in Section~\ref{s-con}, we give some directions for future work.

\section{Preliminaries}\label{s-prelim}

In order to prove our results we will
need two known lemmas.
The first lemma is used in Section~\ref{s-fix} and generalizes Theorem~\ref{t-kb} from SMP to SRP and is implicitly used in the literature (see, for example,~\cite{BBGKP13}); we give its proof for completeness.
The second lemma, which we will use in Sections~\ref{s-fix} and~\ref{s-core}, is due to Letchford, Reinelt and Theis~\cite{LRT08jdm}.

\begin{lemma}\label{compatible}
If $(M,p)$ is a stable solution for an instance $(G,1,w)$ of \emph{SRP}, then $M$ is a maximum weight matching and $p^t$ is a core allocation of the corresponding matching game.
Moreover, for every other maximum weight matching $\hat{M}$ of $G$ there exists an equivalent payoff vector~$\hat{p}$ of $p$ that is compatible with $\hat{M}$.
\end{lemma}

\begin{proof}
Let $(M,p)$ be a stable solution for an instance $(G,1,w)$ of {\sc SRP}. Then the total payoff vector~$p^t$ is readily seen to be in the core of the corresponding matching game (see also~\cite{BKP12,EK01}). Hence $w(M)=p^t(N)=v(N)$ holds, which means that $M$
must be a maximum weight matching.

Let $\hat{M}$ be another maximum weight matching.
As $p^t$ is in the core, we find that
$p^t_i+p^t_j\geq w(ij)$ for every $ij\in \hat{M}$. As $w(\hat{M})=v(N)=p^t(N)$, this means that
$p^t_i+p^t_j=w(ij)$ for every $ij\in \hat{M}$, and moreover $p_i^t=0$ if $i$ is not incident to an edge of $\hat{M}$. Consequently, the vector
$\hat{p}$, defined by $\hat{p}(i,j)=p_i^t$ if $ij\in \hat{M}$ and $\hat{p}(i,j)=0$ otherwise, is equivalent to $p$ and compatible with $\hat{M}$.
\qed
\end{proof}

\begin{lemma}[\cite{LRT08jdm}]\label{lrt}
For a graph~$G$ with vertex capacity function $b$ and edge weighting~$w$,
it is possible to find a maximum weight $b$-matching of $G$ in $O(n^2m\log (n^2/m))$ time.
\end{lemma}

\section{Stable Fixtures with Payments}\label{s-fix}

We start with the following useful lemma. This lemma shows that a matching $M$ in a stable solution $(M,p)$ for an instance $(G,b,w)$ of SFP is always a maximum weight $b$-matching. Moreover it implies that the core of a multiple partners matching game is nonempty if the corresponding instance $(G,b,w)$ of SFP has a stable solution (the reverse does not hold as we will prove in Section~\ref{s-core}).

\begin{lemma}\label{l-sfp_compatible}
Let $(G,b,w)$ be an instance of {\sc SFP}.
If $(M,p)$ is a stable solution for $(G,b,w)$ then $M$ is a maximum weight $b$-matching and $p^t$ is a core allocation of the corresponding
multiple partners matching game.
\end{lemma}

\begin{proof}
Let $(M,p)$ be a stable solution for an instance $(G,b,w)$ of SFP, where $G=(N,E)$.
Before we prove that $M$ is a maximum weight $b$-matching, we will first show that $p^t$
is a core allocation of the multiple partners matching game $(N,v)$ defined on $(G,b,w)$.
Recall that $p^t$ is defined as $p^t(i)=\sum_{j:ij\in E}p(i,j)$ for every $i\in N$. Let
$S\subseteq N$ be an arbitrary coalition, and let $M'$ be a maximum weight $b$-matching in the subgraph of $G$ induced by $S$.  Then we have that
\[\begin{array}{lcl}
p^t(S) &= &\displaystyle\sum_{i\in S}p^t(i)\\[15pt]
&= &\displaystyle\sum_{i\in S}{\big (}\sum_{j:ij\in M\cap M'} p(i,j) + \sum_{j:ij\in M\setminus M'} p(i,j){\big )}\\[15pt]
&= &\displaystyle\sum_{ij\in M\cap M'}(p(i,j)+p(j,i)) + \sum_{i\in S}\sum_{j:ij\in M\setminus M'} p(i,j)\\[15pt]
&= &\displaystyle\sum_{ij\in M\cap M'}w(ij) + \sum_{i\in S}\sum_{j:ij\in M\setminus M'} p(i,j)\\[15pt]
&\geq  &\displaystyle\sum_{ij\in M\cap M'}w(ij) + \sum_{i\in S}\sum_{j:ij\in M'\setminus M} u_p(i)\\[15pt]
&= &\displaystyle\sum_{ij\in M\cap M'}w(ij) +
\sum_{ij\in M'\setminus M} (u_p(i)+u_p(j))\\[15pt]
&\geq &\displaystyle\sum_{ij\in M\cap M'}w(ij) + \sum_{ij\in M'\setminus M} w(ij)=w(M')=v(S).
\end{array}
\]
The first inequality is valid because of the following reason. If $i$ is unsaturated by $M$ then $u_p(i)=0$ by definition. If $i$ is saturated
by $M$ then $|\{j:ij\in M\setminus M'\}|\geq |\{j:ij\in M'\setminus M\}|$ and $p(i,j)\geq u_p(i)$ holds for every $j\neq i$ by definition.
The last inequality follows from the stability condition for the pairs not matched by $M$.
We conclude that $p^t$ is in the core.
In particular this means that $w(M)=\sum_ip^t(i)=v(N)$. Hence, $M$ is a maximum weight $b$-matching of $G$.\qed
\end{proof}

The remainder of the section is organized as follows.
In Section~\ref{s-chsfp} we present our characterizations of the set of stable solutions for a given instance of SFP in terms of  stable solutions for some corresponding instance of SRP and in terms of integral optimal solutions of some appropriate LP relaxation.
In Section~\ref{s-gen3} we first observe that our characterization of the set of stable solutions directly leads to a polynomial-time algorithm for solving SFP and then we present an alternative but asymptotically faster algorithm (in order to do so we use Lemma~\ref{lrt}).
In Section~\ref{s-gen2} we generalize Lemma~\ref{compatible} from SRP to SFP
and simultaneously Theorem~\ref{sot92} from MPA to SFP (apart from one statement of Theorem~\ref{sot92} that we observe cannot be generalized in this way).

\subsection{Characterizing Stable Solutions of SFP}\label{s-chsfp}

In this subsection we show
how stable solutions of an instance $(G,b,w)$ of SFP
correspond to both stable solutions of an
instance $(G',1,w')$ of SRP, where $G'$ is a graph of size $O(n^3)$, and
to integral optimal solutions of an LP relaxation.

\bigskip
\noindent
{\bf Reducing to SRP.}
We first explain how to construct the instance $(G',1,w')$ of SRP from an instance $(G,b,w)$ of SFP; see also Figure~\ref{f-con}.
Our construction is based on a well-known construction, which was introduced by Tutte \cite{Tutte54} for nonbipartite graphs with no edge weights.
We write $G'=(N',E')$.
For each player $i\in N$ with capacity $b(i)$ we create $b(i)$ copies, $i^1, i^2, \ldots, i^{b(i)}$ in $N'$. For each edge $ij\in E$ we create four players, $\overline{i_j}$, $i_j$, $j_i$, $\overline{j_i}$, with edges $i^s\overline{i_j}$ for $s=1,\ldots, b(i)$, $\overline{i_j}i_j$, $i_jj_i$, $j_i\overline{j_i}$ and $\overline{j_i}j^t$ for $t=1, \ldots, b(j)$, each with weight $w(ij)$.
This completes the construction.
Note that $G'$ is bipartite if and only if $G$ is bipartite.
Hence, our construction also reduces an instance of MPA to an instance of SMP.
We say that $(G',1,w')$ is {\it reduced} from $(G,b,w)$.

\begin{figure}
\begin{center}
\scalebox{0.6}{\includegraphics{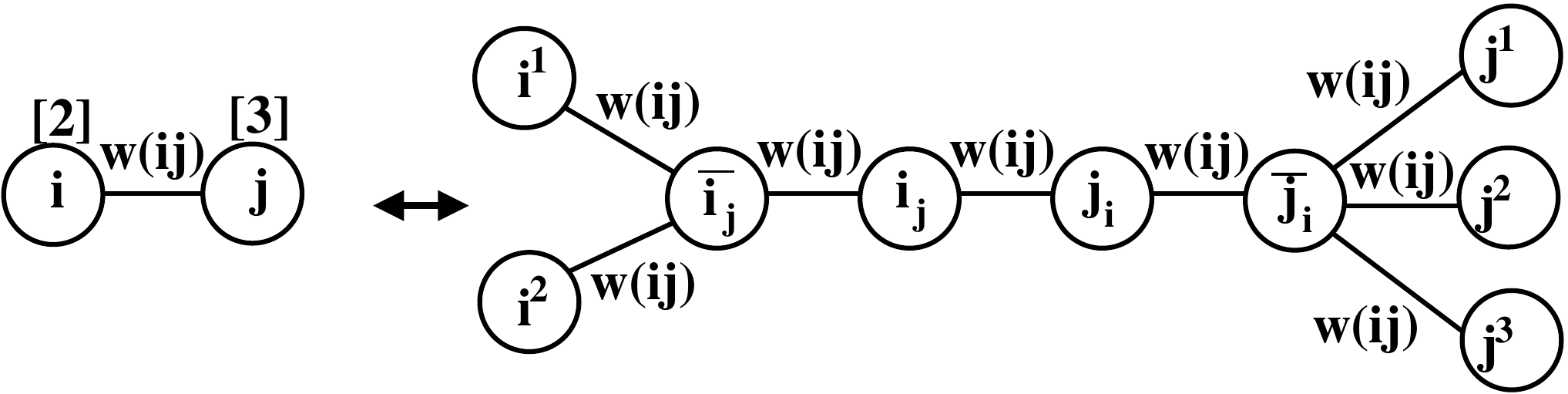}}
\caption{An example of the construction of an SRP instance $(G',1,w')$ from a SFP instance $(G,b,w)$, where $b(i)=2$ and $b(j)=3$.}\label{f-con}
\end{center}
\end{figure}

Given a solution $(M,p)$ for $(G,b,w)$, we define
a solution $(M',p')$ for $(G',1,w')$ as follows.
As $b\equiv 1$ in the instance $(G',1,w')$, we define $p'$ as a function on $N'$ for simplicity, that is, with slight abuse of notation we define $p'$ as a  total payoff vector.
\begin{itemize}
\item
The payoffs of the copies will be the same as the minimum payoffs of the original players, that is, for each $i\in N$, let
$p'(i^s)=u_p(i)$ for every $s=1,\ldots, b(i)$.
\item
For each $i\in N$, we order its partners under $M$.
\begin{itemize}
\item For each $ij\in M$, if $j$ is the  $s$-th partner of $i$ for some $s\in \{1,\ldots, b(i)\}$, and $i$ is the $t$-th partner of $j$ for some $t\in \{1,\ldots, b(j)\}$ then let $i^s\overline{i_j}\in M'$, $i_jj_i\in M'$ and $\overline{j_i}j^t\in M'$ with the following payoffs: $p'(i_j)=p(i,j)$ and $p'(\overline{i_j})=w(ij)-u_p(i)$, and similarly $p'(j_i)=p(j,i)$ and $p'(\overline{j_i})=w(ij)-u_p(j)$.
\item For each $ij\in E\setminus M$, let $\overline{i_j}i_j$, $j_i\overline{j_i}\in M'$
 with $p'(i_j)=\min\{u_p(i),w(ij)\}$ and $p'(\overline{i_j})=w(ij)-\min\{u_p(i),w(ij)\}$, and similarly $p'(j_i)=\min\{u_p(j),w(ij)\}$ and $p'(\overline{j_i})=w(ij)-\min\{u_p(j),w(ij)\}$.
\end{itemize}
\end{itemize}
We say that $(M',p')$ is {\it reduced} from $(M,p)$ and prove in our next lemma that $(M',p')$ is a solution for $(G',1,w)$.
Note that $(M',p')$ may not be the
unique solution reduced from $(M,p)$,
as it might be possible to order the partners of a player differently and we used this order in the construction of $M'$.
We may also define a matching $M'$ in the above way, without defining any payoff vector~$p'$, and say that $M'$ is {\it reduced} from~$M$.
In our next lemma we also show that our reduction preserves the ``maximum weight property'' of a matching.
\begin{lemma}\label{l-obs}
Let $(G,b,w)$ be an instance of {\sc SFP} and let $(G',b',w')$ be its reduced instance. Then the following two statements hold:
\begin{itemize}
\item [(i)] every matching $M'$ reduced from a maximum weight $b$-matching $M$ is a maximum
weight matching of $G'$;
\item [(ii)] every pair $(M',p')$ reduced from a solution $(M,p)$ for $(G,b,w)$ is a solution for $(G',1,w)$.
\end{itemize}
\end{lemma}

\begin{proof}
Let $(G,b,w)$ be an instance of {\sc SFP} and let $(G',b',w')$ be its reduced instance.
We first prove (i) and then (ii).

\medskip
\noindent
{\bf (i)}
Let $M$ be a maximum weight $b$-matching of $G$, and let $M'$ be a matching in $G'$ reduced from $M$.
Note that $$w'(M')=3\sum_{ij\in M}w(ij)+2\sum_{ij\in E\setminus M}w(ij)=3w(M)+2w(E\setminus M)=w(M)+2w(E).$$
Consider a maximum weight matching $M^*$ of $G'$.
Let $E_{ij}$ denote all the edges in $G'$ that are incident with at least one vertex in $\{\overline{i_j}, i_j, j_i, \overline{j_i}\}$.
As $M^*$ has maximum weight, each $E_{ij}$ contains at least two edges of~$M^*$. By construction of~$G'$, each $E_{ij}$ contains
at most three edges of~$M^*$.
If $|E_{ij}\cap M^*|=2$, then we replace the edges of $E_{ij}\cap M^*$ by the edges $\overline{i_j}i_j$ and $j_i\overline{j_i}$ if necessary.
If $|E_{ij}\cap M^*|=3$, then $M^*$ must contain the edge $i_jj_i$ together with the edges $i^s\overline{i_j}$ for some $s\in \{1, \ldots, b(i)\}$ and $\overline{j_i}j^t$ for some $t\in \{1, \ldots, b(j)\}$. Hence the resulting matching $\hat{M}'$ is a maximum weight matching of $G'$, which we can reduce
from some $b$-matching $\hat{M}$ of $G$ by construction. In particular, we have that
$w'(\hat{M}') =w(\hat{M})+2w(E)$.
As $M$ is a maximum weight $b$-matching of $G$, we find that $w(M)\geq w(\hat{M})$.
Hence, we deduce that $w'(\hat{M}') =w(\hat{M})+2w(E) \leq w(M)+2w(E)=w'(M')$. As
$\hat{M}'$ is a maximum weight matching of $G'$, this implies that
$M'$ is a maximum weight matching of $G'$.

\medskip
\noindent
{\bf (ii)} Let $(M,p)$ be a solution for $(G,b,w)$, and let $(M',p')$ be a pair reduced from $(M,p)$.
Every edge in $G'$ is obtained from an edge $ij$ of $G$. We consider all cases (cf. Remark~\ref{r-bis1}).
If $i^s\overline{i_j}\in M'$ then $p'(i^s)+p'(\overline{i_j})= u_p(i)+w(ij)-u_p(i)=w(ij)=w'(i^s\bar{i_j})$.
If $\overline{i_j}i_j \in M'$ then $p'(\overline{i_j})+p'(i_j)=w(ij)-\min\{u_p(i),w(ij)\}+\min\{u_p(i),w(ij)\}=w(ij)=w'(\overline{i_j}i_j)$.
If  $i_jj_i\in M'$ then $p'(i_j)+p'(j_i)=p(i,j)+p(j,i)=w(ij)=w'(i_jj_i)$, where the one-but-last equality follows from the assumption that $(M,p)$ is a solution for $(G,b,w)$.
Hence, we have shown statement~(ii).
\qed
\end{proof}

Our next result shows that reduced solutions preserve stability.
In fact, we prove a stronger statement, as we characterize stable solutions for instances of SFP in terms of stable solutions for instances of SRP.
In particular, Theorems~\ref{t-stt}~(ii) and~(iii) specify properties of stable solutions of reduced instances that we will need later on.
We note that the restriction in statements~(ii) and~(iii) of Theorem~\ref{t-stt} to maximum weight $b$-matchings does not make these statements weaker due to Lemma~\ref{l-sfp_compatible}.

\begin{theorem}\label{t-stt}
Let $(G,b,w)$ be an instance of \emph{SFP} and $(G',1,w')$ be the instance of
\emph{SRP} reduced from $(G,b,w)$.
Then the following three statements hold:
\begin{itemize}
\item [(i)] $(G,b,w)$ has a stable solution if and only if $(G',1,w')$ has a stable solution;\\[-6pt]
\item [(ii)] $(G,b,w)$ has a stable solution $(M,p)$
for some maximum weight $b$-matching~$M$ of $G$
 if and only if $(G',1,w')$ has a stable solution $(M',p')$,
 where $M'$ is reduced from $M$ and $p'(i_j)=p(i,j)$ and $p'(j_i)=p(j,i)$ for every
$ij\in M$;\\[-6pt]
\item [(iii)] a solution $(M,p)$ for $(G,b,w)$,
where $M$ is a maximum weight $b$-matching of $G$,
is stable if and only if every solution $(M',p')$ for $(G',1,w')$ reduced from $(M,p)$ is stable.
\end{itemize}
\end{theorem}

\begin{proof}
Let $(G,b,w)$ be an instance of SFP and $(G',1,w')$ be the instance of SRP reduced from $(G,b,w)$.

\medskip
\noindent
We first prove the ``$\Rightarrow$''-direction of statement~(iii), namely
that if $(G,b,w)$ has a stable solution $(M,p)$, then every $(M',p')$ for $(G',1,w')$ that is reduced from $(M,p)$ is stable.
Note that this immediately implies the ``$\Rightarrow$''-direction of statements~(i) and~(ii) as well.

Let $(M',p')$ be reduced from $(M,p)$. By Lemma~\ref{l-obs}~(ii) we find that $(M',p')$ is a solution for $(G',1,w')$.
We will show that $(M',p')$ is a stable solution for $(G',1,w')$. Throughout the proof we assume that $p'$ is a total payoff vector defined on $N'$.

In order to prove that $(M',p')$ is stable we must check whether the stability condition for each edge not in $M'$ is satisfied, that is, whether
$u_{p'}(i')+u_{p'}(j')\geq w'(i'j')$ holds for all $i'j'\in E'\setminus M'$.
As $u_{p'}(i')=p'(i')$ for every $i'\in N'$,
this comes down to checking whether
\[p'(i')+p'(j')\geq w'(i'j')\; \mbox{for all}\; i'j'\in E'\setminus M'.\]
First consider each edge $i^s\overline{i_j} \notin M'$.
Then $p'(\overline{i_j})=w(ij)-\min\{u_p(i),w(ij)\}$. Hence we find that  $p'(i^s)+p'(\overline{i_j}) = u_p(i)+w(ij)-\min\{u_p(i),w(ij)\}\geq w(ij)=w'(i^s\overline{i_j})$.
Now consider each edge $\overline{i_j}i_j \notin M'$. Then $p'(\overline{i_j})=w(ij)-u_p(i)$ and $p'(i_j)=p(i,j)$.
As $u_p(i)\leq p(i,j)$, we obtain $p'(\overline{i_j})+p'(i_j)
= w(ij)-u_p(i)+p(i,j) \geq w(ij)=w'(\overline{i_j}i_j)$.
Finally, consider each edge $i_jj_i \notin M'$. Then $p'(i_j)=\min\{u_p(i),w(ij)\}$ and $p'(j_i)=\min\{u_p(j),w(ij)\}$.
If $w(ij)\leq u_p(i)$ or $w(ij)\leq u_p(j)$, then $p'(i_j)+p'(j_i)\geq w(ij)=w'(i_jj_i)$.
Now suppose that $w(ij)>u_p(i)$ and $w(ij)>u_p(j)$.
Then $p'(i_j)+p'(j_i)=u_p(i)+u_p(j)\geq w(ij)=w'(i_jj_i)$, where the inequality follows from the fact that $(M,p)$ is a stable solution for $(G,b,w)$.
Hence the ``$\Rightarrow$''-directions of statements (i) and~(ii) have been proven.

\medskip
\noindent
We will now prove the ``$\Leftarrow$''-directions of statements~(i) and~(ii). We do this as follows. First we assume that $(G',1,w')$ has a stable solution. We then show that this stable solution can be changed into a stable solution for $(G',1,w')$ as prescribed by statement~(ii), followed by a proof of the
 ``$\Leftarrow$''-direction of statement~(ii).

Let $(M'',p'')$ be a stable solution for $(G',1,w')$.
Let $M'$ be a matching of $G'$ that can be reduced
from a maximum weight $b$-matching $M$ of $G$.
By Lemma~\ref{l-obs}~(i), we find that $M'$ is a maximum weight matching of $G'$
By Lemma~\ref{compatible},
$(G',1,w')$ has a stable solution $(M',p')$, where $p'$ is equivalent to $p''$.
We define $p(i,j)=p'(i_j)$ and $p(j,i)=p'(j_i)$ for every $ij\in M$ and $p(i,j)=p(j,i)=0$ for every $ij\in E\setminus M$.
We claim that $(M,p)$ is a stable solution for $(G,b,w)$.

First we show that $(M,p)$ is a solution for $(G,b,w)$. Let $ij\in E$. If $ij\in E\setminus M$, then $p(i,j)=p(j,i)=0$ by definition of $p$.
If $ij\in M$, then $i_jj_i\in M'$ by construction.
As $(M',p')$ is a solution, this means that $p(i,j)+p(i,j)=p'(i_j)+p'(j_i)=w'(i_jj_i)=w(ij)$.

Now we show that $(M,p)$ is stable.
For contradiction, assume that $G$ has an edge $ij\in E\setminus M$ that blocks $(M,p)$, that is,
$$u_p(i)+u_p(j)<w(ij).$$
Note that $i_jj_i\in E'\setminus M'$ by construction.
We claim that $p'(i_j)\leq u_p(i)$ and $p'(j_i)\leq u_p(j)$ hold. Then $p'(i_j)+p'(j_i)\leq u_p(i)+u_p(j)<w(ij)=w'(i_jj_i)$, which means. together with $i_jj_i\in E'\setminus M'$, that we have obtained a contradiction with the stability of $(M',p')$.

In order to prove that $p'(i_j)\leq u_p(i)$ and $p'(j_i)\leq u_p(j)$, we distinguish between two cases.
First suppose that $i$ is not saturated by $M$. Then $G'$ contains a copy of $i$, say~$i^s$, that is not matched by $M'$, which means that $p'(i^s)=0$.
As $(M',p')$ is stable and $i^s\overline{i_j}\in E\setminus M'$, we find that $u_{p'}(i^s)+u_{p'}(\overline{i_j})=p'(i^s)+p'(\overline{i_j})\geq w'(i^s\overline{i_j})$.
Hence, $p'(\overline{i_j})\geq w(ij)$.
As $\overline{i_j}i_j\in M'$ by construction and $(M',p')$ is a solution for $(G',1,w)$, we obtain $p'(\overline{i_j})+p'(i_j)=w'(\overline{i_j}i_j)$.  As $p'(\overline{i_j})\geq w(ij)$ and
$p'\geq 0$, this means
 that $p'(i_j)=0$. Hence, $p'(i_j)\leq u_p(i)$.

Now suppose that $i$ is saturated by $M$. Then all  $b(i)$ copies $i^h$ of $i$ in $(G',1,w')$ are matched by $M'$. Let $k$ be the partner of $i$ for which $p(i,k)$ is minimal, so $u_p(i)=p(i,k)=p'(i_k)$. As $i^s\overline{i_k}\in M'$ and $(M',p')$ is a solution for $(G',1,w')$, we find that $p'(i^s)+p'(\overline{i_k})=w(ik)$. From the stability of $(M',p')$ it follows that $p'(\overline{i_k})+p'(i_k)\geq w(ik)$. Hence $p'(i_k)\geq p'(i^s)$ holds.
From the stability of $(M',p')$ it also follows that $p'(i^s)\geq w(ij)-p'(\overline{i_j})$. As $(M',p')$ is a solution for $(G',1,w')$ and $\overline{i_j}i_j\in M'$, we find that
$w(ij)-p'(\overline{i_j})=p'(i_j)$. Putting the last three conditions together yields $p'(i_j)\leq p'(i_k)=u_p(i)$, just as we claimed.
We can show $p'(j_i)\leq u_p(j)$ by exactly the same arguments.
Hence the two ``$\Leftarrow$''-directions of statements (i) and (ii) have been proven.

\medskip
\noindent
We are left to prove the ``$\Leftarrow$''-direction of statement (iii). This direction can be seen as follows.
Let $(M,p)$ be a solution for $(G,b,w)$, where $M$ is a maximum weight $b$-matching of $G$.
If the starting stable solution $(M'',p'')$ is reduced from a solution $(M,p)$ of $(G,b,w)$ (instead of being an arbitrary stable solution for $(G',1,w)$) then the above argument shows that $(M,p)$ is stable. This completes the proof of the theorem.
\qed
\end{proof}

Note that not all stable solutions for $(G',1,w')$ can be reduced from stable solutions for $(G,b,w)$, as $u_p(i)>p'(i^s)$ is possible for some $s\in 1\ldots b(i)$. This might be the case even for very small instances. Let $G$ consist of only two adjacent players $i$ and $j$ with $b(i)=b(j)=1$ and $w(ij)=7$.
Then the pair $(M,p)$ with $M=\{ij\}$ and $p(i,j)=3$, $p(j,i)=4$ is a stable solution, which has only one reduced solution $(M',p')$ given by
$M'=\{i^1\overline{i_j}, i_jj_i, \overline{j_i}j^1\}$ and payoffs $p'(i^1)=3$, $p'(\overline{i_j})=4$, $p'(i_j)=3$, $p'(j_i)=4$, $p'(\overline{j_i})=3$ and $p'(j^1)=4$. An alternative stable solution for $(G',1,w')$ is given by $(M',p'')$, where $p''(i^1)=1$, $p''(\overline{i_j})=6$, $p''(i_j)=3$, $p''(j_i)=4$, $p''(\overline{j_i})=4$, $p''(j^1)=3$.

\begin{remark}\label{fourtotwo}
In the proof of Theorem~\ref{t-stt} we could have used a simpler construction, namely the one where each 4-player path in the reduced instance, $\overline{i_j},i_j,j_i,\overline{j_i}$ is replaced with a 2-player path $\overline{i_j},\overline{j_i}$.
However, this reduction would not give us the one-to-one correspondence described in
Theorem~\ref{t-stt}~(ii), and this correspondence is crucial for some of our other results.
\end{remark}

\bigskip
\noindent
{\bf The LP Relaxation.}
Let $(G,b,w)$ be an instance of SFP.
Let $M$ be a matching of $G$.
With $M$ we associate a binary vector $x^M:E\to \{0,1\}$ called the {\it characteristic function} of $M$, which is
defined by $x^M(ij)=1$ for all $ij\in M$ and $x^M(ij)=0$ for all $ij\in E\setminus M$.  Then we can write $\sum_{j:ij\in E} x^M(ij)\leq b(i)$ for each $i\in N$ as an alternative way to state the capacity condition. This leads to the following LP relaxation, which we call Primal-$(G,b,w)$:

\begin{equation}\tag{P-obj}
\label{P_obj}
\max \sum_{ij\in E}w(ij)x(ij)
\end{equation}
\begin{equation}\tag{a}
\label{a}
\mbox {subject to}\;\;\; \sum_{j:ij\in E}x(ij)\leq b(i) \mbox{ for each } i\in N,
\end{equation}
\begin{equation}\tag{b}
\label{b}
\hspace*{18mm}0\leq
x(ij)\leq 1 \mbox{ for each } ij\in E.
\end{equation}
We make the following observation.

\begin{lemma}\label{l-primal}
The (optimal) solutions of  the integer linear program corresponding to Primal-$(G,b,w)$ correspond to the (maximum weight) $b$-matchings of $G$ and vice versa.
\end{lemma}
We now formulate the dual LP Dual-$(G,b,w)$ of Primal-$(G,b,w)$:

\begin{equation}\tag{D-obj}
\label{D_obj}
\min \sum_{i\in N}b(i)y(i)+\sum_{ij\in E}d(ij)
\end{equation}
\begin{equation}\tag{a'}
\label{a'}
\mbox{subject to}\;\;\;y(i)+y(j)+d(ij)\geq w(ij) \mbox{ for each } ij\in E,
\end{equation}
\begin{equation*}\tag{b'}
\label{b'}
0 \leq y(i) \mbox{ for all } i\in N,
\end{equation*}
\begin{equation*}\tag{c'}
\label{c'}
\hspace*{3mm}0 \leq d(ij) \mbox{ for all } ij\in E.
\end{equation*}
Note that for an optimal dual solution $(y,d)$, it holds that $d(ij)=[w(ij)-y(i)-y(j)]_+$ (where the latter notation means $\max\{w(ij)-y(i)-y(j), 0\}$).

\medskip
\noindent
We now characterize stable solutions of instances of SFP in terms of integral optimal solutions of Primal-$(G,b,w)$.

\begin{theorem}\label{char}
An instance $(G,b,w)$ of \emph{SFP}
has a stable solution if and only if the LP Primal-$(G,b,w)$ has an integral optimal solution.
\end{theorem}

\begin{proof}
Let $(G,b,w)$ be an instance of SFP.
By Theorem~\ref{t-stt}~(i) it suffices to show that the instance $(G',1,w')$ of SRP reduced from SFP has a stable solution if and only if
Primal-$(G,b,w)$ has an integral optimal solution. We show this
below.~\footnote{Note that we could have proven Theorem~\ref{char} directly, without  reducing to an instance of SRP. However, as a byproduct of our proof
we obtain an alternative description of the set of stable solutions of the reduced instance  (Corollary~\ref{c-bijection2}), which we think is worth noting.}

\medskip
\noindent
{\bf ($\Rightarrow$)}
Suppose that $(M'',p'')$ is a stable solution for $(G',1,w')$, where we again assume that $p''$ is a total payoff vector
(defined on $N'$).
Let $M'$ be a maximum weight matching of $G'$ that can be reduced from a maximum weight $b$-matching $M$ of $G$. By Lemma~\ref{compatible},
$(G',1,w')$ has a stable solution $(M',p')$, where $p'$ is equivalent to $p''$.

Before we define a feasible solution $(y,d)$ of Dual-$(G,b,w)$, we will first show that $p'(i^s)=p'(i^t)$ for every player $i\in N$ and every two indices
$s,t\in \{1,\ldots, b(i)\}$. For contradiction, suppose that $p'(i^s)>p'(i^r)$ for some $i\in N$ and indices  $r,s\in \{1,\ldots, b(i)\}$. Then, as
$p'(i^r)\geq 0$, we have that $p'(i^s)>0$. This means that  $i^s\overline{i_j} \in M'$ for some $j\in N\setminus \{i\}$, which implies that
$p'(i^s)+ p'(\overline{i_j}) =w'(i^s\overline{i_j})=w(ij)$. Consequently, $p'(i^r)+ p'(\overline{i_j})<p'(i^s)+ p'(\overline{i_j})=w(ij)=w'(i^r\overline{i_j})$,
thus $i^r\overline{i_j}$ would block $(M',p')$. This contradicts the fact that $(M',p')$ is a stable solution for $(G',1,w)$.
Hence we may set
\begin{equation}
\label{y}
y(i)=p'(i^s)\; \mbox{for every copy}\; i^s\;
\mbox{of
every}\; i\in N,
\end{equation}
as this is well-defined.
We also define
\begin{equation}
\label{d}
d(ij)=[w(ij)-y(i)-y(j)]_+\; \mbox{for
every}\;
ij\in E.
\end{equation}
In this way we have obtained
our feasible solution $(y,d)$ of Dual-$(G,b,w)$.

We define an integral solution $x$ of Primal-$(G,b,w)$ as follows. For each $ij\in E$, let $x(ij)=1$ if $i_jj_i\in M'$ and $x(ij)=0$ otherwise.
Recall that $i^s\overline{i_j}\in M'$ for exactly one index $s\in \{1,\ldots,b(i)\}$ if and only if $i_jj_i\in M'$. Hence, as each $i^s$ is incident to at most one edge of $M'$,
we find that $\sum_{j:ij\in E}x(ij)=\sum_{j:i_jj_i\in M'}1\leq b(i)$. As $0\leq x(ij)\leq 1$ for every $ij\in E$, we find that $x$ is a feasible solution.

It remains to prove that $x$ is an optimal solution of Primal-$(G,b,w)$.
In fact we will show
that $x$ and $(y,d)$ are both optimal,
as they satisfy the complementary slackness
conditions:
\begin{itemize}
\item For every $i\in N$,
if $\sum_{j:ij\in E}x(ij)< b(i)$ then
(as $i^s\overline{i_j}\in M'$ for exactly one index $s\in \{1,\ldots,b(i)\}$ if and only if $i_jj_i\in M'$)
some copy of $i$, say $i^r$, is unmatched in $M'$. Therefore $y(i)=p'(i^r)=0$.
This means that dual condition~(\ref{b'}) is binding. \\[-8pt]
\item
For every $ij\in E$, if $x(ij)<1$, then $x(ij)=0$ and thus $i_jj_i\notin M'$ . Then, by construction of $M'$,
we have that  $i^s\overline{i_j}\notin M'$ for
any
$s\in \{1,\ldots, b(i)\}$,
 $\overline{i_j}i_j\in M'$, $j_i\overline{j_i}\in M'$ and $\overline{j_i}j^t\notin M'$ for
any
$t\in \{1,\ldots, b(j)\}$. These facts, together with the stability conditions
for  $i^s\overline{i_j}$ and $\overline{j_i}j^t$ and the fact that $y(i)=p'(i^s)$ holds by definition, imply that
\[\begin{array}{lcl}
y(i)+p'(\overline{i_j}) = p'(i^s)+p'(\overline{i_j})
\geq w'(i^s\overline{i_j})
= w(ij) = w(\overline{i_j}i_j)
=p'(i_j)+p'(\overline{i_j}),
\end{array}\]
and thus $y(i)\geq p'(i_j)$, and similarly, $y(j)\geq p'(j_i)$, which means that
 $y(i)+y(j)\geq p'(i_j) +p'(j_i)\geq w'(i_jj_i)=w(ij)$, where
the second inequality is implied by the stability condition for $i_jj_i$. As a consequence, $d(ij)=[w(ij)-y(i)-y(j)]_+=0$
holds.
This means that dual condition~(\ref{c'}) is binding. \\[-8pt]
\item
For every $ij\in E$,
if $x(ij)>0$, then $x(ij)=1$ and thus $i_jj_i\in M'$
and therefore, by construction of $M'$,
we have that  $i^s\overline{i_j}\in M'$ for some $s\in \{1,\ldots, b(i)\}$,
 $\overline{i_j}i_j\notin M'$, $j_i\overline{j_i}\notin M'$ and $\overline{j_i}j^t\in M'$ for some $t\in \{1,\ldots, b(j)\}$. These facts, together with the stability conditions for  $\overline{i_j}i_j$ and $j_i\overline{j_i}$ and the fact that $y(i)=p'(i^s)$ holds by definition, imply that
\[\begin{array}{lcl}
y(i)+p'(\overline{i_j}) = p'(i^s)+p'(\overline{i_j})
= w'(i^s\overline{i_j})
= w(ij) = w(\overline{i_j}i_j)
\leq p'(i_j)+p'(\overline{i_j}),
\end{array}\]
and thus $y(i)\leq p'(i_j)$, and similarly, $y(j)\leq p'(j_i)$, which means that
$y(i)+y(j)\leq p'(i_j)+p'(j_i)=w'(i_jj_i)=w(ij)$;
as a consequence, $d(ij)=[w(ij)-y(i)-y(j)]_+=w(ij)-y(i)-y(j)$
holds. This means that dual condition~(\ref{a'}) is binding.
\end{itemize}
Hence $x$ is an
integral optimal
solution of Primal-$(G,b,w)$, as required.

\medskip
\noindent
{\bf ($\Leftarrow$)}
Let $x$ be an integral optimal solution for Primal-$(G,b,w)$
and let $(y,d)$ be an optimal solution of Dual-$(G,b,w)$.
From $x$ and $(y,d)$ we create a stable solution $(M',p')$ for $(G',1,w')$ as follows.

As $x$ is an integral optimal solution, $x$ is a characteristic function of some maximum weight $b$-matching $M$ by Lemma~\ref{l-primal}. Let $M'$ be
a
maximum weight matching in $G'$ reduced from $M$. For every $i\in N$, let $p'(i^s)=y(i)$ for each $s\in \{1, \ldots , b(i)\}$.
For each $ij\notin M$
we define $p'(i_j)=y(i)$ and $p'(\overline{i_j})=w(ij)-y(i)$.
For each $ij\in M$ we choose $\xi(i,j) \ge 0$ and $\xi(j,i) \ge 0$ with $\xi(i,j)+\xi(j,i) =d(ij)$ and define $p'(i_j) = y(i) + \xi(i,j)$ and $p'(\overline{i_j})=w(ij)-y(i)$.

We will first show that $(M',p')$ is a solution for $(G',1,w')$, that is, $p'(i')+p'(j')=w'(i'j')$ for every $i'j'\in M'$.
Suppose $ij\in M$. Then $i^s\overline{i_j}\in M'$ for some $s=1,\ldots,b(i)$ and $j^t\overline{j_i}\in M'$ for some $t=1,\ldots,b(j)$.
By symmetry, we only need to consider the edge $i^s\overline{i_j}$. We find that $p'(i^s)+p'(\overline{i_j})=y(i)+w(ij)-y(i)=w(ij)=w'(i^s\overline{i_j})$.
We also have $i_jj_i\in M'$ and find that
$$p'(i_j)+p'(j_i) = y(i) + \xi(i,j)+y(j) + \xi(j,i)=y(i)+y(j)+d(ij)=w(ij)=w'(i_jj_i),$$ where the one-but-last equality follows from the fact that
dual condition~(\ref{a'}) is binding for $ij\in M$, as $x(ij)=1>0$.

Now suppose $ij\notin M$. Then $\overline{i_j}i_j\in M'$ and $\overline{j_i}j_i\in M'$. By symmetry we only need to consider the edge $\overline{i_j}i_j$.
We find that $p'(\overline{i_j})+p'(i_j)=w(ij)-y(i)+y(i)=w(ij)=w'(\overline{i_j}i_j)$.
Hence, $p'(i')+p'(j')=w'(i'j')$ for every $i'j'\in M'$, which means that $(M',p')$ is a solution for $(G',1,w')$.

Now we show that $(M',p')$ is stable, that is, we show that $p'(i')+p'(j')\geq w'(i'j')$ for every $i'j'\in E'\setminus M'$.
First suppose $ij\in M$. Then $\overline{i_j}i_j\in E'\setminus M'$ and $p'(\overline{i_j})+p'(i_j)=w(ij)-y(i)+y(i)=w(ij)=w'(\overline{i_j}i_j)$, as required.
We also have $\overline{j_i}j_i\in E'\setminus M'$ and can show that the stability condition is satisfied for $\overline{j_i}j_i$ in the same way.

Now suppose $ij\in E\setminus M$. Then $i_jj_i\in E'\setminus M'$.
As $x(ij)=0<1$, dual condition~(\ref{c'}) is binding for $ij$, so $d(ij)=0$.
Consequently, we deduce that
$$p'(i_j)+p'(j_i)=y(i)+y(j)=y(i)+y(j)+d(ij)\geq w(ij)=w'(i_jj_i),$$
where the last inequality
follows from the fact that $(y,d)$ is a solution of Dual-$(G,b,w)$.
We also have $i^s\overline{i_j}\notin E'\setminus M'$ for $s=1,\ldots,b(i)$ and $j^t\overline{j_i}\notin E'\setminus M'$ for $t=1,\ldots,b(j)$.
By symmetry it suffices to consider an edge $i^s\overline{i_j}$ for some $s\in \{1,\ldots,b(i)\}$. Then we obtain
$p'(i^s)+p'(\overline{i_j})=y(i)+w(ij)-y(i)=w(ij)=w'(i^s\overline{i_j})$.
We conclude that the stability condition is satisfied for every $i'j'\in E'\setminus M'$.
This completes our proof.\qed
\end{proof}

The following corollary is obtained directly from the proof of Theorem~\ref{char}.
It gives a full description of the set of stable solutions of every instance of SFP that is reduced from some instance of SRP.
In particular, this explains why this set may contain stable solutions that cannot be reduced from stable solutions for $(G,b,w)$ (recall that we gave an explicit example of such a stable solution in the paragraph after the proof of Theorem~\ref{t-stt}).

\begin{corollary}\label{c-bijection2}
Let $(G',1,w')$ be an instance of \emph{SRP} reduced from an instance $(G,b,w)$ of \emph{SFP} that has a stable solution.
Then $(G',1,w')$ has a stable solution $(M',p')$
 if and only if  Dual-$(G,b,w)$ has an optimal solution $(y,d)$ with $y(i)=p'(i^s)$ for every $i\in N$ and $s\in \{1, \ldots , b(i)\}$
 (and $d(ij)=[w(ij)-y(i)-y(j)]_+$ for every $ij\in E$).
\end{corollary}

\begin{proof}
Let $(G',1,w')$ be an instance of SRP reduced from an instance $(G,b,w)$ of SFP that has a stable solution.

\medskip
\noindent
{\bf ($\Rightarrow$)} Suppose that $(G',1,w')$ has a stable solution $(M',p')$.
Then  Dual-$(G,b,w)$ has an optimal solution $(y,d)$ with $y(i)=p'(i^s)$ for every $i\in N$ and $s\in \{1, \ldots , b(i)\}$
 and $d(ij)=[w(ij)-y(i)-y(j)]_+$ for every $ij\in E$ due to equations~(\ref{y}) and~(\ref{d}) in the proof of Theorem~\ref{char}.

\medskip
\noindent
{\bf ($\Leftarrow$)}
Suppose that Dual-$(G,b,w)$ has an optimal solution $(y,d)$ with $y(i)=p'(i^s)$ for every $i\in N$ and $s\in \{1, \ldots , b(i)\}$.
Recall that the equalities $d(ij)=[w(ij)-y(i)-y(j)]_+$ $(ij\in E)$ follow directly from the optimality of $(y,d)$.
As $(G,b,w)$ has a stable solution,  Theorem~\ref{char} tells us that  Primal-$(G,b,w)$ has an integer optimal solution~$x$.
In the proof of Theorem~\ref{char} we created a stable solution $(M',p')$ from $x$
and $(y,d)$. \qed
\end{proof}

We finish Section~\ref{s-chsfp} by summarizing our main results in the following corollary (the $1\Leftrightarrow 2$ proof follows directly from
Theorem~\ref{t-stt}~(i), whereas $2\Leftrightarrow 3$ follows directly from
Theorem~\ref{char}).

\begin{corollary}\label{c-char}
Let $(G,b,w)$ be an instance of \emph{SFP} and $(G',1,w')$ be
the
instance of \emph{SRP} reduced from $(G,b,w)$.
The following statements are equivalent:
\begin{enumerate}
\item $(G,b,w)$ has a stable solution.
\item $(G',1,w')$ has a stable solution.
\item Primal-$(G,b,w)$ has an integral optimal solution.
\end{enumerate}
\end{corollary}

\subsection{Solving SFP efficiently}\label{s-gen3}

We first show how to solve SFP on an instance $(G,b,w)$ by using the results from Section~\ref{s-chsfp}.
We construct the instance $(G',1,w')$ of SRP. This takes
$O(n^3)$ time as we may assume without loss of generality that $b(i)\leq n$ for all $i\in N$ and thus
$|N'|=\sum_{i\in N}b(i)+4m=O(n^2)$ and $|E'| \leq \sum_{i\in N}b(i)n+3m\leq n^3 +3m= O(n^3)$.
We then use the aforementioned algorithm of Bir\'o et al.~\cite{BKP12} to
compute in $O(n'm'+n'^2\log n')=O(n^5)$ time a stable solution for $(G',1,w')$ or else conclude that $(G',1,w')$ has no stable solution.
In the first case we can modify the stable solution into a stable solution for $(G,b,w)$ in $O(n^3)$ time, as described in the proof of Theorem~\ref{char}. In the second case, Theorem~\ref{char} tells us that $(G,b,w)$ has no stable solution. The total running time is $O(n^5)$.
Below we present an algorithm that solves SFP in $O(n^2m \log (n^2/m))$ time.

A \emph{half-$b$-matching} in a graph $G=(N,E)$ with an integer vertex capacity function $b$ and an edge weighting $w$
is an edge mapping $f$ that maps each edge $e$ to a value in
$\{0,\frac{1}{2},1\}$,
such that $\sum_{e:v\in e}f(e)\leq b(v)$ for each $v\in N$. The weight of $f$ is $w(f)=\sum_{e\in E}w(e)f(e)$.

Let $(G,b,w)$ be an instance of SFP. We define its \emph{duplicated} instance $(\hat{G}, \hat{b}, \hat{w})$ of MPA as follows. We replace each player $i$ of $G$ by two players $i'$ and $i''$ in $\hat{G}$ with the same capacities, that is, we set $\hat{b}(i')=\hat{b}(i'')=b(i)$. Moreover, we replace each edge $ij$ by two edges $i'j''$ and $i''j'$ with half-weights, that is, we set $\hat{w}(i'j'')=\hat{w}(i''j')=\frac{1}{2}w(ij)$.

In a previous paper~\cite{BKP12}, three of us proved the following statement for instances of SRP only. We now generalize it  for instances of SFP by using similar arguments.

\begin{theorem}\label{t-ch}
An instance $(G,b,w)$ of \emph{SFP} admits a stable solution if and only if the maximum weight of a $b$-matching in $G$ is equal to the maximum weight of a half-$b$-matching in $G$.
\end{theorem}

\begin{proof}
Suppose that $(\hat{G}, \hat{b}, \hat{w})$ is the duplicated MPA instance of $(G,b,w)$. From
Theorem~\ref{sot92} we know that $(\hat{G}, \hat{b}, \hat{w})$ admits a stable solution $(\hat{M}, \hat{p})$, where $\hat{M}$ is a maximum weight $b$-matching in $\hat{G}$.

Let $\hat{x}$ be the characteristic function of $\hat{M}$.
Let $(\hat{y}, \hat{d})$ denote a dual optimal solution, which has the same objective function value as $\hat{x}$ by Theorem \ref{char}.
We define a half-$b$-matching $f$ for $(G,b,w)$ by setting $f(ij)=\frac{\hat{x}(i'j'')+\hat{x}(i''j')}{2}$ for each $ij\in E(G)$.
We observe that $$w(f)=\sum_{ij\in E(G)}w(ij)f(ij)=\sum_{ij\in E(G)}\hat{w}(i'j'')\hat{x}(i'j'')+\hat{w}(i''j')\hat{x}(i''j')=\hat{w}(\hat{x}).$$
Therefore $f$ is a feasible solution for Primal-$(G,b,w)$ with the same value as $\hat{x}$.
We create a feasible dual solution $(y,d)$ for Dual-$(G,b,w)$ from $(\hat{y}, \hat{d})$ with the
same objective function
value, namely
$\hat{w}(\hat{x})$, as follows. Let $y(i)=\hat{y}(i')+\hat{y}(i'')$ for every player $i$ of $G$ and let $d(ij)=\hat{d}(i'j'')+\hat{d}(i''j')$ for every edge $ij$ in $E(G)$. Let us verify that $(y,d)$ is a feasible dual solution, i.e. that for
every edge $ij\in E(G)$
we have
\[\begin{array}{lcl}
y(i)+y(j)+d(ij) &= &(\hat{y}(i')+\hat{y}(i''))+(\hat{y}(j')+\hat{y}(j''))+(\hat{d}(i'j'')+\hat{d}(i''j'))\\[5pt]
&= &(\hat{y}(i')+\hat{y}(j'')+\hat{d}(i'j''))+(\hat{y}(i'')+\hat{y}(j')+\hat{d}(i''j'))\\[5pt]
&\geq &\hat{w}(i'j'')+\hat{w}(i''j')\\[5pt]
&= &w(ij).
\end{array}
\]
Therefore $f$ is a feasible solution for Primal-$(G,b,w)$ and $(y,d)$ is a feasible solution for Dual-$(G,b,w)$ with the same
value, namely
$\hat{w}(\hat{x})$, so by the Weak Duality theorem both $f$ and $(y,d)$ are optimal solutions.
Hence we get that $f$ is indeed a maximum weight half-$b$-matching with total weight equal to the primal-dual optimum for both $(G,b,w)$ and $(\hat{G}, \hat{b}, \hat{w})$. It is now readily seen that $(G,b,w)$ has a stable solution if and only if the maximum weight of a $b$-matching is equal to the maximum weight of a half-$b$-matching of $(G,b,w)$.\qed
\end{proof}

We observe that the maximum weight of a $b$-matching can be computed in time $O(n^2m\log (n^2/m))$, as described in Lemma~\ref{lrt}.
The maximum weight of a half-$b$-matching can be computed in the same
time, since the maximum weight of a half-$b$-matching in $(G,b,w)$ is the same as the maximum weight of a $b$-matching in a duplicated bipartite graph $(\hat{G}, \hat{b}, \hat{w})$, as explained in the
proof of Theorem~\ref{t-ch}.
This leads to the following result.

\begin{theorem}\label{algo}
\emph{SFP} can be solved in $O(n^2m \log (n^2/m))$ time.
\end{theorem}

\subsection{Generalizing Lemma~\ref{compatible} and Theorem~\ref{sot92} to SFP}\label{s-gen2}

We will now completely generalize Lemma~\ref{compatible} from SRP to SFP.
In this way we will also extend corresponding statements in Theorem~\ref{sot92} from MPA to SFP (see also Section~\ref{s-existencestable}). Note that there exist instances of SRP (for instance, take a triangle) and thus of SFP with no stable solution. Hence, it is not possible to give a full generalization of Theorem~\ref{sot92},
as the statement that each instance of MPA has a stable solution cannot be generalized.
For proving our result we need one additional lemma.

\begin{lemma}\label{l-lastpart}
Let $(G,b,w)$ be an instance of \emph{SFP} that has a stable solution $(M,p)$, and
let $(G',1,w')$ be the instance of \emph{SRP} reduced from $(G,b,w)$. Let $(M',p')$ be a solution for $(G',1,w')$ reduced from $(M,p)$.
Then for every maximum weight $b$-matching $\hat{M}$ of~$G$, the pair $(\hat{M},\hat{p})$, where
\begin{itemize}
\item $\hat{p}(i,j)=p'(i_j)$ and $\hat{p}(j,i)=p'(j_i)$ for $ij\in \hat{M}$
\item $\hat{p}(i,j)=\hat{p}(j,i)=0$ for $ij\in E\setminus \hat{M}$,
\end{itemize}
is a stable solution for $(G,b,w)$. Moreover, $\hat{p}$ is equivalent to $p$.
\end{lemma}

\begin{proof}
Let $(G,b,w)$ be an instance of {SFP} that has a stable solution $(M,p)$.
Let $(G',1,w')$ be the instance of {SRP} reduced from $(G,b,w)$.
Moreover, let $(M',p')$ be a solution for $(G',1,w')$ reduced from $(M,p)$.
Now let $\hat{M}$ be a maximum weight $b$-matching of $G$, and let
$\hat{p}$ be the pay-off vector defined by $\hat{p}(i,j)=p'(i_j)$ and $\hat{p}(j,i)=p'(j_i)$ for $ij\in \hat{M}$  and
$\hat{p}(i,j)=\hat{p}(j,i)=0$ for $ij\in E\setminus \hat{M}$.
We must prove that $(\hat{M},\hat{p})$ is a stable solution for $(G,b,w)$.
Consider a matching $\hat{M}'$ of~$G'$ that is reduced from~$\hat{M}$.
By Lemma~\ref{l-obs}~(i), we find that $\hat{M'}$ is a maximum weight matching of $G'$.
As $(M',p')$ is a stable solution of $(G',1,w')$ by Theorem~\ref{t-stt}~(iii), we may apply Lemma~\ref{compatible} to find that $(\hat{M}',p')$ is a stable solution for $(G',1,w')$ as well
(recall that we view $p'$ as a total payoff vector defined on $N'$).
By Theorem~\ref{t-stt}~(ii) we find that $(\hat{M},\hat{p})$ is a stable solution for $(G,b,w)$.
We are left to prove that $p$ and $\hat{p}$ are equivalent, that is, we show the following four conditions:
\begin{itemize}
\item $u_p(i)=u_{\hat{p}}(i)$ for every $i\in N$,
\item $p(i,j)=\hat{p}(i,j)$ and $p(j,i)=\hat{p}(j,i)$  for every $ij\in M \cap \hat{M}$,
\item $p(i,j)=u_p(i)=u_{\hat{p}}(i)$ and $p(j,i)=u_p(j)=u_{\hat{p}}(j)$ for every $ij\in M\setminus\hat{M}$, and
 \item $\hat{p}(i,j)=u_{\hat{p}}(i)=u_p(i)$ and $\hat{p}(j,i)=u_{\hat{p}}(j)=u_p(j)$ for every $ij\in \hat{M}\setminus M$.
\end{itemize}
First suppose that $ij\in M \cap \hat{M}$. Then $p(i,j)=p'(i_j)=\hat{p}(i,j)$.
Now suppose that $ij\in \hat{M}\setminus M$.
Then $ij\in E\setminus M$. By definition, this means that $\overline{i_j}i_j\in M'$. As $(M',p')$ is a solution for $(G',1,w')$, we find that
$p'(\overline{i_j})+p'(i_j)=w'(\overline{i_j}i_j)$. As $ij\in \hat{M}$ and $\hat{M'}$ is reduced from $\hat{M}$, we also find by definition that $i^s\overline{i_j}\in \hat{M}'$ for some $s\in \{1,\ldots,b(i)\}$. 
As $(\hat{M}',p')$ is a solution for $(G',1,w')$, this means that $p'(i^s)+p'(\overline{i_j})=w(ij)$. We conclude that
$\hat{p}(i,j)=p'(i_j)=p'(i^s)=u_p(i)$, where the last equality follows from the definition of $p'(i^s)$.
The same arguments show that for $ij \in M\setminus\hat{M}$ we have $p(i,j)=p'(i_j)=p'(i^s)=u_p(i)$.
Also, for every $i\in N$ we deduce that
\[\begin{array}{lcl}
u_{\hat{p}}(i) &= &\displaystyle\min_{ik\in \hat{M}}\hat{p}(i,k)\\[10pt]
&= &\displaystyle\min_{ik\in \hat{M}}p'(i_k)\\[10pt]
&=&\displaystyle\min{\big\{}\min_{ik\in \hat{M}\cap M}p'(i_k),\min_{ik\in \hat{M}\setminus M}p'(i_k)\big\}\\[10pt]
&\geq &\displaystyle \min\big\{\min_{ik\in M}p'(i_k),\min_{ik\in \hat{M}\setminus M}p'(i_k)\big\} \\[10pt]
&= &\displaystyle \min\big\{\min_{ik\in M}p(i,k),\min_{ik\in \hat{M}\setminus M}p'(i_k)\big\} \\[10pt]
&= &\displaystyle\min \big\{  u_p(i),\min_{ik\in \hat{M}\setminus M}p'(i_k) \big \} \\[10pt]
&=&\min \big \{ u_p(i),u_p(i) \big\} \\[10pt]
&= &u_p(i),
\end{array}\]
where the one-but-last equality uses the fact that  $p'(i_k)=u_p(i)$ for every $ik\in \hat{M}\setminus M$.
If $u_p(i)=0$, then the facts that $p(i,j)=\hat{p}(i,j)$ for every $ij\in M \cap \hat{M}$ and $\hat{p}(i,j)=u_p(i)$ for every $ij\in \hat{M}\setminus M$ imply that
$u_p(i)=u_{\hat{p}}(i)=0$. If $u_p(i)>0$, then $u_{\hat{p}}(i)\geq u_p(i)>0$. Hence, $i$ is saturated by $\hat{M}$, which implies that $u_p(i)=u_{\hat{p}}(i)$ as well.
Consequently, for $ij \in M\setminus\hat{M}$ we find that $p(i,j)=u_p(i)=u_{\hat{p}}(i)$.
Hence,  for all $i\in N$, we find that $u_p(i)=\min_{j:ij\in E}p(i,j)=\min_{j:ij\in E}\hat{p}(i,j)=u_{\hat{p}}(i)$.
We conclude that $p$ and $\hat{p}$ are indeed equivalent.
This completes the proof of the lemma.
\qed
\end{proof}

We are now ready to present the generalization of Lemma~\ref{compatible} from SRP to SFP and appropriate parts of Theorem~\ref{sot92} from MPA to SFP.

\begin{theorem}\label{sfp_compatible}
Let $(G,b,w)$ be an instance of {\sc SFP}.
If $(M,p)$ is a stable solution for $(G,b,w)$ then $M$ is a maximum weight $b$-matching and $p^t$ is a core allocation of the corresponding
multiple partners matching game.
Moreover, for every other maximum weight $b$-matching $\hat{M}$ of $G$ there exists an equivalent payoff vector~$\hat{p}$ of $p$ that is compatible with $\hat{M}$.
\end{theorem}

\begin{proof}
The first statement follows immediately from Lemma~\ref{l-sfp_compatible}.
The second statement follows immediately from Lemma~\ref{l-lastpart}.\qed
\end{proof}

\section{Connecting MPA with SMP}\label{s-alternative}

In this section we show how Theorem~\ref{t-stt} readily implies the three main results of Sotomayor for MPA in~\cite{Sotomayor92},~\cite{Sotomayor99} and~\cite{Sotomayor07}, respectively. Essentially what we do in each of our three alternative proofs is reducing MPA to SMP, that is, to the one-to-one case, which then enables us to apply classical results for SMP. These alternative proofs are therefore not just
simpler than the originals, but also shed light on the connection between
MPA and SMP.

Sotomayor obtained her results for MPA in the context of economic markets. Consider a two-sided market, consisting of
two disjoint groups, namely
a group $I$ of {\it sellers} and a group  $J$ of {\it buyers} ($I$ and $J$ are also called the {\it buyer} and {\it seller side} of the market, respectively).
 The set $E$ consists of pairs $(i,j)$ with $i\in I$ and $j\in J$ that may go in business with each other. Hence,
we can define a bipartite graph $G=(I\cup J,E)$ where $I$ and $J$ are the two partition classes.
Each seller~$i\in I$ has $b(i)$ identical objects (and hence can sell at most $b(i)$ objects).
Each buyer~$j\in J$ may purchase at most one object from each seller and can buy at most $b(j)$ objects in total.
So, the vector $b$ can be seen as a vertex capacity vector.
For each $(i,j)\in E$, we denote the gain of a transaction between seller $i$ and buyer $j$ by $w(ij)$.
This yields an instance $(G,b,w)$ of MPA.

\subsection{The Existence of a Stable Solution}~\label{s-existencestable}

First we provide an alternative proof for Theorem~\ref{sot92}, which we restate below.

\medskip
\noindent
{\bf Theorem~\ref{sot92}~(\cite{Sotomayor92}).}
{\it Every instance $(G,b,w)$ of {\sc MPA} has a stable solution, which can be found in polynomial time.
If $(M,p)$ is a stable solution then $M$ is a maximum weight $b$-matching and $p^t$ is a core allocation of the corresponding
multiple partners
assignment game.
Moreover, for every other maximum weight $b$-matching $\hat{M}$ of $G$ there exists an equivalent payoff vector~$\hat{p}$ of $p$ that is compatible with $\hat{M}$.}

\begin{proof}
Combining Theorems~\ref{t-kb} and~\ref{t-stt} immediately implies the existence of a stable solution for every instance of {\sc MPA}.
The other parts of the theorem are covered by Theorem~\ref{algo} (the polynomial-time result on finding a stable solution) and Theorem~\ref{sfp_compatible} (the remaining part).
\qed
\end{proof}

\subsection{The Lattice Structure of Stable Solutions}

A {\it complete lattice} is a partially ordered set in which every subset has a supremum (obtained via join operations) and an infimum (obtained via meet operations).
Sotomayor \cite{Sotomayor99} proved the following result on the lattice structure of the payoff vectors of stable solutions of an instance of MPA,
for which we give a simple alternative proof.
Here, a stable solution $(M,p)$ of an instance $(G,b,w)$ is {\it optimal} for a set $S$ if there is no stable solution $(M',p')$ with
$\sum_{j:ij\in E}p'(i,j) > \sum_{j:ij\in E}p(i,j)$ for some $i\in S$.

\begin{theorem}[Sotomayor \cite{Sotomayor99}]\label{sot99}
The payoff vectors of stable solutions of an instance of {\sc MPA} form a complete lattice with unique optimal outcomes for
 each side of the market.
\end{theorem}

\begin{proof}
Shapley and Shubik~\cite{SS72} proved that the payoff vectors of stable solutions of an instance of SMP form a complete lattice with unique optimal outcomes for each side of the market.
They defined the meet and joint operations for two payoff vectors $p$ and $q$ by setting $[p\vee q](i)=\min\{p(i),q(i)\}$ and $[p\wedge q](i)=\max\{p(i),q(i)\}$ for every $i\in I$, and $[p\vee q](j)=\max\{p(j),q(j)\}$ and $[p\wedge q](i)=\min\{p(j),q(j)\}$ for every $j\in J$.

We combine the result of Shapley and Shubik with Theorem~\ref{t-stt} after making the following additional remarks.
By Theorem~\ref{t-stt}~(ii)
there is a one-to-one correspondence between the stable solutions
for an instance $(G,b,w)$ of SFP and the payoffs of some of the players in its reduced instance $(G',1,w')$ of SRP. Namely, a pair $(M,p)$ is a stable solution for $(G,b,w)$ if and only if there exists a stable solution $(M',p')$ for $(G',1,w')$ with $p'(i_j)=p(i,j)$ and $p'(j_i)=p(j,i)$ for every $ij\in M$.

We define meet and joint operations for the set of stable payoff vectors for $(G,b,w)$. Let $p$ and $q$ be two payoff vectors of two stable solutions for $(G,b,w)$ and let $p'$ and $q'$ be the payoff vectors of the reduced stable solutions for $(G',1,w')$, respectively. For each edge $ij\in M$, where $i\in I$ and $j\in J$, let
\begin{itemize}
\item $[p\vee q](i,j)=\min\{p(i,j),q(i,j)\}=\min\{p'(i_j),q'(i_j)\}=[p'\vee q'](i_j)$,
\item $[p\wedge q](i,j)=\max\{p(i,j),q(i,j)\}=\max\{p'(i_j),q'(i_j)\}=[p'\wedge q'](i_j)$,
\item $[p\vee q](j,i)=\max\{p(j,i),q(j,i)\}=\max\{p'(j_i),q'(j_i)\}=[p'\vee q'](j_i)$,
\item $[p\wedge q](j,i)=\min\{p(j,i),q(j,i)\}=\min\{p'(j_i),q'(j_i)\}=[p'\wedge q'](j_i)$.
\end{itemize}
We need to show that $\vee$ and $\wedge$ are well defined operators on the set of stable payoff vectors for $(G,b,w)$. Let us consider operation $\vee$ (operation $\wedge$ can be treated in the same way). By
Theorem~\ref{t-stt}~(ii)
we can transform $p$ and $q$ to $p'$ and $q'$, respectively, where $p'$ and $q'$ are stable payoffs for $(G',1,w')$. Then we create $p'\vee q'$ as Shapley and Shubik defined, which is a stable payoff vector for $(G',1,w')$. Finally we construct $p\vee q$ as defined above for $(G,b,w)$. This is a payoff vector of a stable solution for $(G,b,w)$ by
Theorem~\ref{t-stt}~(ii).

Let $S'\subset N'$ be the set of players of form $i_j$ and $j_i$ in $G'$. Shapley and Shubik~\cite{SS72} proved that the payoff vectors of the stable solutions for $(G',1,w')$ form a lattice, therefore the restrictions of the stable payoff vectors on $S'$ also form a lattice. This lattice is equivalent to the lattice  of the payoff vectors of the stable solutions for $(G,b,w)$,
as we explained above. This also implies the existence of unique optimal outcomes for each side of the market for $(G,b,w)$.
\qed
\end{proof}

\subsection{Competitive Equilibrium Outcomes}

In economic theory, competitive equilibrium outcomes in MPA  form a well studied solution concept. For instance they are used in multi-unit auction mechanisms (see, for example,~\cite{M00}).
Course allocation, where students bid for courses with virtual money, is  another interesting application;
such a system has been recently introduced at Wharton University~\cite{BK14}.
In order to explain competitive equilibrium outcomes, we first need to extend our terminology.

Let $(G,b,w)$ be an instance of MPA.
Recall that $G=(I\cup J,E)$ is a bipartite graph with partition classes $I$ (the set of sellers) and $J$ (the set of buyers).
Each player $i$ is endowed with $b(i)$ identical, player-specific goods.
Let $q(i^s)$ denote the price of the $s$th item of seller $i\in I$. This yields a
{\it price vector}~$q$.
Recall that every buyer and seller can only make one transaction of a single item between them
and that they may only do so if they are connected by an edge.
For $ij\in E$, the weight $w(ij)$ expresses the mutual benefit when seller~$i$ and buyer~$j$ make such a transaction: if $j$ buys the $s$th item of $i$ then the {\it payoff} for $i$ is
 $p(i,j,s)=q(i^s)$ and the {\it payoff} for $j$  is $p(j,i,s)=w(ij)-q(i^s)$.
This leads to a vector~$p$, which is called the {\it payoff vector} with respect to $q$.

For a buyer $j\in J$, a {\it bundle} $B(j)$ is a set of items that can be bought by $j$ and that satisfies the following conditions:
\begin{itemize}
\item [(i)] $B(j)$ contains at most one item of every seller;
\item [(ii)] $B(j)$ contains at most $b(j)$ items.
\end{itemize}
For a bundle $B(j)$ and a price vector $q$, the {\it total payoff} for $j$ is the sum of the payoffs $p(j,i,s)$ over all items in $B(j)$.
The {\it demand set} $D_q(j)$ of buyer $j$ is
the set of all bundles of $B(j)$
that maximize the total payoff for $j$ over all bundles of $j$.

Let $B=(B(1),\ldots,B(|J|))$ be a vector, where $B(j)$ is a bundle for buyer $j$ for $j=1,\ldots,|J|$, and let $q$ be a price vector. Then $(B,q)$ is called a {\it competitive equilibrium outcome} for $(G,b,w)$ if the following three conditions hold:
\begin{itemize}
\item [(i)] for every buyer $j\in J$,
$B(j)\in D_q(j)$ holds;
\item [(ii)] every unsold copy of an item has zero price, that is, $q(i^s)=0$ if the $s$th item of buyer $i$ is not in some $B(j)$;
\item [(iii)] the sets $B(1),\ldots, B(|J|)$ are pairwise disjoint.
\end{itemize}
Note that this definition implies that for every competitive equilibrium outcome~$(B,q)$, we must have that $q(i^s)=q(i^t)$ for every $i,s,t$, that is, prices of (identical)
items
of each seller must be the same:
if two items of the same seller would have different prices, that is, if for some $i,j,s,t$, $q(i^s)>q(i^t)$,
then in particular, $q(i^s) >0$, so $i^s$ must be sold
due to condition (ii).
However, if $j$ buys the $s$th item of $i$, that is, $i^s \in B(j)$, then $B(j)\notin D_q(j)$
violating condition~(i).
Hence $(B,q)$ would not be a competitive equilibrium
outcome for $(G,b,w)$.

A competitive equilibrium
outcome $(B,q)$ for an instance $(G,b,w)$ of MPA
is readily seen to correspond to a stable solution $(M,p)$, where
 $p$ is the payoff vector with respect to $q$, where $p(i,j)=p(i,k)$ for every $i\in I$ and $j,k\in J$, and $p(i,j)=0$ if $i$ is unsaturated in $M$.
Note that we could have defined the competitive equilibrium outcomes simply as stable solutions where for every player on one side of the market her payoffs are identical in every pair she is involved in and equal to zero if she is unsaturated. However, we decided to deduce this property from the general definition of competitive equilibrium outcomes, just as Sotomayor did in \cite{Sotomayor07}, to explain why this well-known notion can be characterized by this simple property for MPA.

The main result of Sotomayor \cite{Sotomayor07} is the following
theorem, for which we give a simple alternative proof.

\begin{theorem}[Sotomayor \cite{Sotomayor07}]\label{sot07}
The set of competitive equilibrium outcomes of
an instance of {\sc MPA}
forms a complete lattice with unique optimal outcomes for each side of the market. The
competitive equilibrium outcomes are stable solutions.
Moreover, the buyer-optimal stable solution is equivalent to the buyer-optimal competitive equilibrium outcome.
\end{theorem}

\begin{proof}
Again we
reduce an instance $(G,b,w)$ of MPA to an instance $(G',1,w')$ of SMP, as described in Theorem~\ref{t-stt}.
For any stable solution $(M,p)$ of $(G,b,w)$ let $(M',p')$ be the reduced stable solution
for $(G',1,w')$. Now we can create another stable solution $(M',\hat{p}')$ for $(G',1,w')$, where for every $ij\in M$ we set $\hat{p}'(i_j)=u_p(i)\leq p'(i_j)=p(i,j)$
and $\hat{p}'(j_i)=w_{ij}-u_p(i)\geq p(j,i)$ and
keep the other payoffs the same.
Since we only modified the redistribution between $i_j$ and $j_i$ (that is, $j_i$ is now getting all the potential surplus in this pair), the only potential blocking pair
for $\hat{p}'$ is
the pair $\{i_j,\overline{i_j}\}$.
However, this pair cannot be blocking since $\hat{p}'(\overline{i_j})=p'(\overline{i_j})=w_{ij}-u_p(i)$, so
$\hat{p}'(\overline{i_j})+\hat{p}'(i_j)=w(ij)$.

Let $(M,\hat{p})$ be the solution for $(G,b,w)$ from which $(M',\hat{p}')$ is reduced.
Here, every seller
$i$ receives identical payoffs from each of her buyers, that is, for
every
$ij\in M$ we have $\hat{p}(i,j)=u_p(i)$. Therefore this solution is a competitive equilibrium
outcome $(M,q)$, where the price of
the items of each seller~$i$ is $q(i)=u_p(i)$
(and the corresponding bundles $B(j)$ are given by the items matched to $j$ in $M$).
Moreover, the price vectors of the set of competitive equilibrium outcomes form a lattice, which is a sublattice of the lattice of the payoff vectors of the stable
solutions. This is because the meet and joint operations (as defined in the proof of Theorem~\ref{sot99}) result in price vectors of competitive equilibrium outcomes
when applied to price vectors of competitive equilibrium outcomes
(indeed, if all copies of an item have identical prices with respect to two given stable payoffs $p$ and $\hat{p}$, then this also holds for their meet and join).

Finally, we show that the buyer-optimal stable solution is equivalent to the buyer-optimal competitive equilibrium outcome.
Let $(M,p)$ be a buyer-optimal solution (where $p$ is uniquely determined). The above construction yields, via $p'$ and $\hat{p}$, a stable solution $(M,\hat{p})$,
which is a competitive equilibrium. By construction, we have $\hat{p}'(i,j) = u_p(i) \le p(i,j)$. However, since $p$ was buyer optimal, a strict inequality is impossible.
Hence $\hat{p}=p$ must hold. \qed
\end{proof}

\section{Core Properties}\label{s-core}

To get started we recall that the core of a multiple partners matching game may be empty, as we could take
for instance a triangle on three vertices and set  $b\equiv 1$ and $w\equiv 1$.
Below we show that this may also be the case for $b\neq 1$
by presenting the following example, which shows that the core of a multiple partners matching game may
be empty even if
$b\equiv \alpha$ for an arbitrary constant $\alpha\geq 2$.

\medskip
\noindent
{\it Example 4.}
Let $\alpha\geq 2$. We start by taking a cycle on three vertices $s_1$, $s_2$ and $s_3$.
For $i=1,\ldots,3$, we add $\alpha-1$ pendant (degree~1) vertices $t_i^1,\ldots,t_i^\alpha-1$ to $s_i$; see also Figure~\ref{f-four}.
We set $b\equiv \alpha$ and $w\equiv 1$.
Then $v(N)=3(\alpha-1)+1=3\alpha-2$. We may assume without loss of generality that $p(t_i^j)=0$ for $i=1,2,3$ and $j=1,\ldots,\alpha-1$.
Then, by symmetry, $p(s_i)=\alpha-\frac{2}{3}$ for $i=1,2,3$. Take the coalition $S=\{s_1,s_2,t_1^1,\dots,t_1^{\alpha-1},t_2^1,\ldots,t_2^{\alpha-1}\}$. It holds that $p(S)=2\alpha-\frac{4}{3}<2\alpha-1 = v(S)$ and thus  the core of this game is empty.

\begin{figure}
\begin{center}
\scalebox{0.55}{\includegraphics{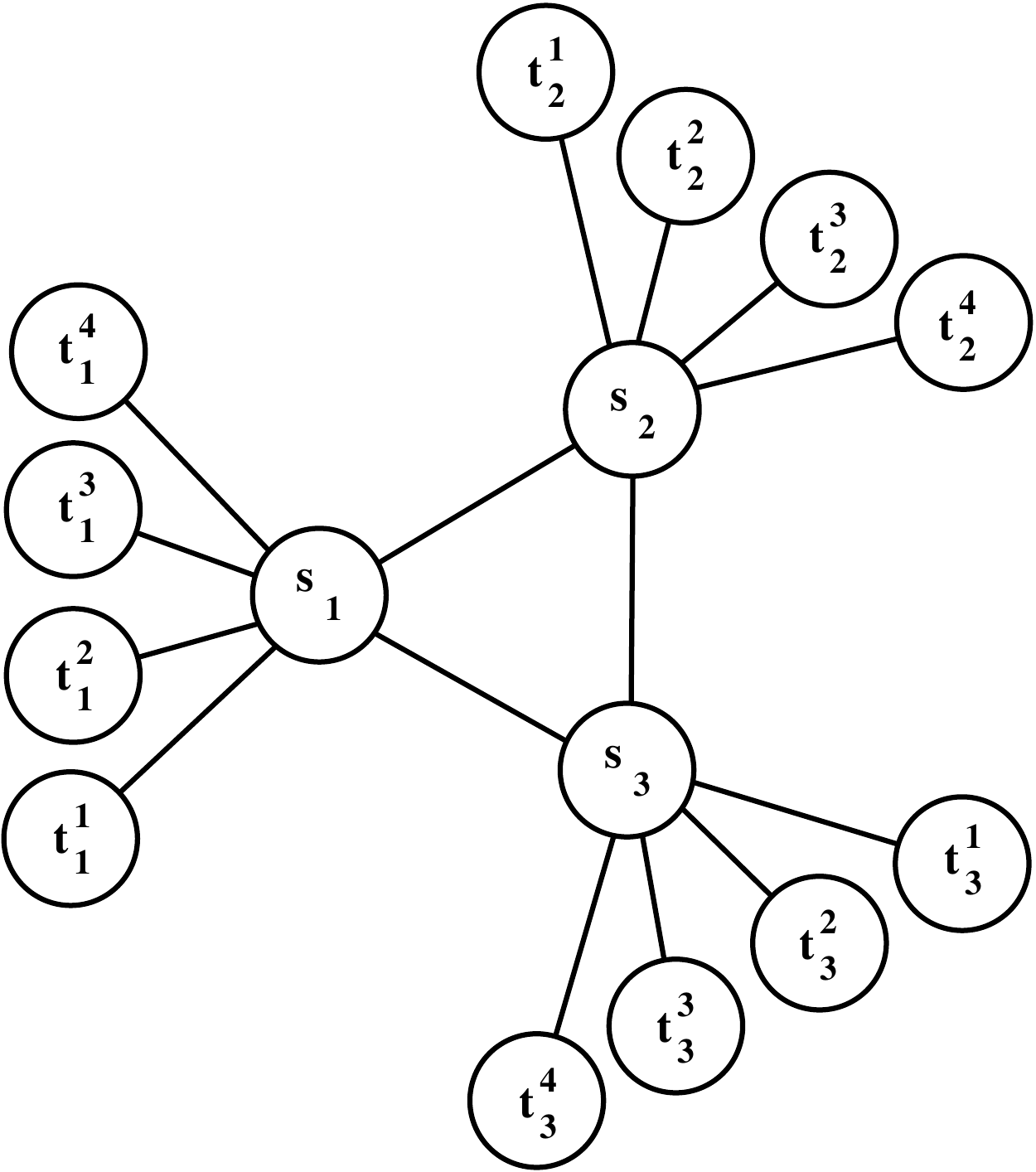}}
\caption{The graph from Example~4 for $\alpha=5$.}\label{f-four}
\end{center}
\end{figure}

To increase our insight into the core of a multiple partners matching game, we now show that every core allocation $x$ of a multiple partners matchings game equals the total payoff vector~$p^t$ of a suitable payoff vector~$p$, which we can find in polynomial time. More precisely, we may even fix an arbitrary maximum weight $b$-matching~$M^*$ to obtain suitable payoffs $p(i,j)$ from the edges of $M^*$.
For proving this, we first present a small lemma which basically states that  a corresponding system of equations (without any non-negativity constraints) is solvable. For a matching $M$ of a graph $G=(N,E)$, we let $N(M)$ denote the subset of vertices of $G$ incident with at least one edge of $M$.

\begin{lemma}\label{lem:solvep}
Let $(N,v)$ be a multiple partners matching game on some triple $(G,b,w)$.
Let $M^*$ be a maximum weight $b$-matching of $G$, and let $x \in \R^N$ be a vector satisfying
 $x(N(M^*_k))=w(M^*_k)$) for every connected component  $M^*_k$ of $M^*$. Then it is possible to find in polynomial time a vector $p$ with (not necessarily positive) entries $p(i,j)$ such that
 \[\begin{array}{llll}
 p(i,j)+p(j,i) = w(ij) &&\text{for~all~} {ij}\in M^*\\
 p(i,j)=p(j,i) = 0 &&\text{for~all~} {ij} \in E\backslash M^*\\
 p^t(i):=\sum_{j:ij\in E}p(i,j)=x(i)  &&\text{for~all~} i \in N.
 \end{array}
 \]
 \end{lemma}

\begin{proof}
Suppose $M^*$ contains a cycle~$C$. We do as follows. We pick an edge $ij\in C$ and define $p(i,j)=p(j,i)=\frac{1}{2}w(ij)$. Then we decrease $x(i)$ and $x(j)$ by the same amount $\frac{1}{2}w(ij)$. We remove the edge $ij$ and decrease $b(i)$ and $ b(j)$ by one each. In this way, $M^* \backslash \{ij\}$ is a maximum weight matching of the reduced instance. We repeat this operation on another cycle that contains only edges of $M^*\backslash \{ij\}$ until such cycles no longer exist. As each time we apply the operation we reduce $M^*$ by at least one edge, this procedure terminates in polynomial time.

From the above we may assume that $M^*$ induces a forest.
If $M^*$ induces a tree, we do as follows. Fix a leaf~$i$ of the forest and let $i \in N(M^*)$ be matched by ${ij} \in M^*$. Set $p(i,j):= x(i)$ and $p(j,i):=w(ij) -x(i)$. We set ${x}(j):=x(j)-p(j,i)=x(j)+x(i)-w(ij)$ and perform the same operation on a leaf of $M^*\setminus \{ij\}$ until $M^*$ has no more edges. If $M^*$ induces a forest, we consider each connected component of $M$ separately.  Each time we apply the operation we reduce $M^*$ by at least one edge. Hence this procedure, and thus our whole algorithm, runs in polynomial time.\qed
 \end{proof}

\medskip
Our first theorem in this section can be obtained by applying Lemma \ref{lem:solvep} to core allocations $x$.

\begin{theorem}\label{thm:structure}
Let $(N,v)$ be a multiple partners matching game on some triple $(G,b,w)$.
For every maximum weight $b$-matching $M^*$ of $G$ and every core allocation~$x$ of $(N,v)$, we can find
in polynomial time a payoff vector $p$ compatible with $M^*$, such
that $x=p^t$.
\end{theorem}

\begin{proof}
The core constraints ensure that $x$ satisfies the requirements $x(N(M^*_k))=w(M^*_k)$ for every component $M^*_k$ of $M^*$. Hence we may apply
Lemma~\ref{lem:solvep} to obtain in polynomial time a vector $p$ with (not necessarily positive) entries $p(i,j)$
 such that
 \[\begin{array}{llll}
 p(i,j)+p(j,i) = w(ij) &&\text{for~all~} {ij}\in M^*\\
 p(i,j)=p(j,i) = 0 &&\text{for~all~} {ij} \in E\backslash M^* \\
 p^t(i):=\sum_{j:ij\in E} p(i,j)=x(i)  &&\text{for~all~} i \in N.
 \end{array}
 \]
If $p \ge 0$ then we are done. Suppose not.
Pick an arbitrary edge $ij \in M^*$ with $p(i,j) < 0$.  Let $$A := \{ij\; |\; ij \in M^* \text{~and~} p(j,i)>0\}.$$
Note that $ij \in A$, since $p(j,i)=w(ij)-p(i,j)>w(ij) \ge 0$.
We claim that there exists a directed path of arcs in $A$ that starts in $j$ and that ends in $i$. For contradiction, suppose that such a directed path does not exist. Let $S\subseteq N$ denote the set of vertices that can be reached from $j$ along arcs in $A$. Thus $i \notin S$, $j \in S$,
and $p(k,\ell) \le 0$ for all $k \notin S, \ell \in S$. Note that for $k=i$ and $\ell=j$ we get a strict inequality. Hence, for $\bar{S}:= N\backslash S$ we conclude that
\[\begin{array}{lcl}
x(\bar{S}) &= &\sum_{k \notin S} x(k)\\[6pt]
&= &\sum_{k\notin S} \sum_{\ell:k\ell\in E} p(k,\ell)\\[6pt]
&= &\sum_{k \notin S} (\sum_{\ell\in S:k\ell \in E}p(k,\ell) +\sum_{\ell \notin S:k\ell\in E} p(k,\ell))\\[6pt]
&< &\sum_{k \notin S} \sum_{\ell \notin S:k\ell\in E} p(k,\ell)\\[6pt]
&=&w(M^*_{\bar{S}}),
\end{array}\]
where $M^*_{\bar{S}}$ is the restriction of $M^*$ to edges with both end-vertices in $\bar{S}$ (recall that $p(i,j)=p(j,i)=0$ if $ij\in E\setminus M^*$).
This strict inequality contradicts our assumption that $x$ is a core allocation of $G$.

Due to the above claim we can find (in polynomial time) a directed path $P$ of arcs in $A$ that starts in $j$ and that ends in $i$.
Together with
arc $ij \in A$ this gives us a directed cycle~$C \subseteq A$ such that for each arc $ij \in C$ we can increase $p(i,j)$ and decrease $p(j,i)$
by  $\epsilon = \min_{{ij} \in C} p(j,i)>0$. This reduces the minimal total amount of negative payoffs of $p$. We repeat this operation until $p\geq 0$. Each time we do this, we reduce at least one payoff $p(j,i)$ to zero. Hence, our algorithm runs in polynomial time.\qed
\end{proof}

Theorem~\ref{sfp_compatible} tells us that if $(M,p)$ is a stable solution for an instance $(G,b,w)$ of SFP, then $p^t$ is a core allocation of the corresponding multiple partners matching game. This is in line with corresponding results for the other models and thus shows that the notion of stability is well defined with respect to the core definition.
However, Theorem~\ref{thm:structure} seems to suggest that core allocations may also correspond to solutions that are not stable.
Indeed the following result shows that this may be the case even if the corresponding instance of {\sc SFP} has no stable solutions.

\begin{theorem}\label{t-nonemptyempty}
There exist multiple partners matching games with a non-empty core but for which the corresponding instance of {\sc SFP} has no stable solutions.
\end{theorem}

\begin{proof}
Consider the following example.
Take a {\it diamond}, that is, a cycle on three vertices $s_1, s_2, s_3$ to which we add a fourth vertex $u$ with edges $s_2u$ and $s_3u$.
We set $b(s_i)=2$ for $i=1,2,3$ and $b(u)=1$, and $w\equiv 1$.
Then Theorem~\ref{t-ch} tells us that a stable solution does not exist, since the
maximum weight of a $b$-matching is $3$, whilst the maximum weight of a half-$b$-matching is $3\frac{1}{2}$
(for the latter, take $f(s_1s_2)=f(s_1s_3)=1$ and $f(s_2s_3)=f(s_2u)=f(s_3u)=\frac{1}{2}$).
However, the total payoff vector $p^t$ defined by $p^t(s_i)=1$ for $i=1,\ldots,3$ and $p^t(u)=0$, which
corresponds to, say, the $b$-matching $M=\{s_1s_2, s_1s_3, s_2s_3\}$ with payoffs
$p(s_1,s_2)=1$, $p(s_2,s_3)=1$ and $p(s_3,s_1)=1$ and zero payoffs for the other edges
is in the core.\qed
\end{proof}

Due to Theorem~\ref{t-nonemptyempty}, the analysis of the core cannot be reduced to the the case in which we have unit vertex capacities (in contrast to our  results in Section~\ref{s-fix}). We distinguish between the cases $b\leq 2$ (which includes the known case of matching games, where $b=1$) and
and $b=3$.

\subsection{The Case $b\leq 2$}

\medskip
\noindent
We analyze the case $b(i) \le 2$ for $i=1, ..., n$ as follows. Note that a player $i$ with $b(i)=0$
necessarily gets $0$ payoff in any core allocation, so the problem reduces trivially to $G[N\backslash\{i\}]$. For
this reason we assume $b(i) \ge 1$ for all $i \in N$ in the following.
For $b\leq 2$ we can provide a positive answer to problem~P3, which implies that testing core membership can be done in polynomial time.
Recall that a positive answer to P3 implies the existence of a polynomial-time algorithm that either finds that the core is empty, or else obtains a core allocation.
As mentioned in Section~\ref{s-intro}, our algorithm uses an algorithm that solves the tramp steamer problem (which we formally define below) as a subproblem.

Let $G=(N,E)$ be a graph with an edge weighting
$p:E\to \R_+$ called the {\it profit function} and an edge weighting $w:E\to \R_+$ called the {\it cost function}.
Let $C=(N_C,E_C)$ be a
cycle of $G$.
The {\it profit-to-cost ratio} for a cycle $C$ is $\frac{p(C)}{w(C)}$ where we write $p(C)=p(E_C)$ and $w(C)=w(E_C)$.
The tramp steamer problem is that of finding a cycle $C$ with maximum profit-to-cost ratio.
This problem is well-known to be polynomial-time solvable
both for directed and undirected graphs
(see~\cite{L76}, or \cite{M79} for a treatment of more general ``fractional optimization'' problems).

\begin{lemma}\label{l-tramp}
The tramp steamer problem can be solved in polynomial time.
\end{lemma}

We are now ready to prove the following result.

\begin{theorem}\label{t-core}
The problem of deciding whether an allocation~$p$ is in the core of a multiple partners matching game can be solved in polynomial time if $b\leq 2$.
Moreover, if $p$ is not in the core, then a coalition $S$ with $p(S)<v(S)$ can be found in polynomial time as well.
\end{theorem}

\begin{proof}
Let $(N,v)$ be a multiple partners matching game defined on a triple $(G,b,w)$, where $b(i)\leq 2$ for all
$i\in N$. Given $S \subseteq N$, a maximum weight
$b$-matching in $G[S]$ is composed of cycles and paths. Hence the core can be alternatively described by the
following (slightly smaller) set of constraints:

\begin{equation}\label{coreqred}
 ~~~~\begin{array}{rrrl}
  p(C)&\ge&w(C),& {\rm for ~all~ cycles~} C \in \mathcal{C}\\
  p(P)&\ge&w(P),& {\rm for~ all~ paths~~} P \in \mathcal{P}\\
  p(N)&=&v(N).&
 \end{array}
\end{equation}
Here, $\mathcal{C}$ stands for the set of
cycles $C \subseteq E$ in $G$  with
$b(i) = 2$ for all $i \in V(C)$.
Similarly, $\mathcal{P}$ stands for the set of
paths $P$ with
$b(i) = 2$ for all inner points on~$P$.

Given $p \in \R^N$, we can check in polynomial time whether $p(N)=v(N)$ holds by computing a
maximum weight  $b$-matching in $G$, which can be done in polynomial time by
Lemma~\ref{lrt}.
Thus we are left with the inequalities for cycles and paths in~(\ref{coreqred}).

 We deal with the cycles first. Let
 $N_2 := \{i\in N:b(i) = 2\}$ and $G_2 := G[N_2]$. In the induced graph $G_2=(N_2,E_2)$, we
 ``discharge'' the given allocations $p(i)$ to the edges by setting $p(i,j) := (p(i)+p(j))/2$ for every edge $ij$ in $G_2$.
 This defines an edge weighting $p: E_2 \rightarrow \R$ such that the core constraints for cycles are equivalent to
 $$ \max_{C\in \mathcal{C}} ~\frac{w(C)}{p(C)}~ \le~ 1,$$
 where the maximum is taken over all cycles in $G_2$. Hence we obtained an instance of the tramp steamer problem, which is polynomial-time
 solvable by Lemma~\ref{l-tramp}.
 Note that by solving the above minimization problem we either find that all cycle constraints
  in (\ref{coreqred}) are satisfied or we end up with a particular cycle $C$ corresponding to a violated core inequality.
  (The latter is of particular importance if we intend to use the ``membership oracle'' as a subroutine for the ellipsoid
  method.)

In what follows, we thus assume that all cycle constraints in (\ref{coreqred}) are satisfied by the given vector $p \in \R^N$
and turn to the path constraints. We process these separately for all possible endpoints $i_0,j_0 \in N$ (with $i_0 \neq
j_0$) and all possible lengths $k=1, ..., n-1$. Let $\mathcal{P}_k(i_0,j_0) \subseteq \mathcal{P}$ denote the set of
$i_0-j_0$-paths of length $k$ in $G$. We construct a corresponding auxiliary graph $G_k(i_0,j_0)$, a subgraph of
$G\times P_{k+1}$, the product of $G$ with a path of length $k$. To this end,
 let $N_2^{(1)}, ..., N_2^{(k-1)}$ be $k-1$
copies of $N_2$. The vertex set of $G_k(i_0,j_0)$ is then $\{i_0,j_0\} \cup N_1^{(1)} \cup ... \cup N_2^{(k-1)}$. We denote
the copy of $i\in N_2$ in $N_2^{(r)}$ by $i^{(r)}$. The edges of $G_k(i_0,j_0)$ and their weights~$\bar{w}$ can then be defined as
\[
 \begin{array}{l}
 _0j^{(1)} ~~~~~{\rm ~ for ~} i_0j \in E ~~{\rm ~ with~ weight~} \bar{w}(i_0j):=
p(i_0)+p(j)/2-w(i_0j)\\
i^{(r-1)}j^{(r)} {~\rm for ~} ij \in E ~{\rm ~~~ with~ weight~}
\bar{w}(ij)~:=(p(i)+p(j))/2-w(ij)\\
i^{(k-1)}j_0~~ {\rm ~ for ~} ij_0 \in E ~~{\rm ~ with~ weight~} \bar{w}(ij_0):=p(i)/2+p(j_0)-w(ij_0).\\
 \end{array}
 \]
We claim that $p(P) \ge w(P)$ holds for all $P \in \mathcal{P}$ if and only if the ({\em w.r.t.} $\bar{w}$) shortest
 $i_0-j_0$--path in $G_k(i_0,j_0)$ has weight $\ge 0$ for all $i_0 \neq j_0$ and $ k=1, .., n-1$.
 Then what is left to do,  in order to verify whether  $p(P) \ge w(P)$ holds for all $P \in \mathcal{P}$,
is to solve $O(n^3)$ instances of the shortest path problem, each of which have size $O(n^2)$
by using the well-known Bellman-Ford algorithm~\cite{Be58};
as we show below, if this yields a shortest path with weight $<0$ we have immediately found a coalition corresponding to a violated core inequality,

 First suppose some
 $P \in \mathcal{P}$ has $p(P) < w(P)$. Let $i_0$ and $j_0$ denote its endpoints and let $k$ denote its length.
 Then $P\in \mathcal{P}_k(i_0,j_0)$ corresponds to an $i_0-j_0$-path $\bar{P}$ in $G_k(i_0,j_0)$ of weight $\bar{w}(\bar{P})<0.$
Now we will show that $p(P) \ge w(P) $ for all $P \in \mathcal{P}$ implies $\bar{w}(\bar{P}) \ge 0$ for all $i_0-j_0$-paths $\bar{P}$
 in any $G_k(i_0,j_0)$. Indeed, an $i_0-j_0$-path~$\bar{P}$ visiting players $i_0, i_1^{(1)}, ..., i_{k-1}^{(k-1)}, j_0$
 corresponds to a
 $i_0-j_0$ path $P \subseteq E$ in $G$ plus possibly a number of cycles $C_1, ..., C_s \subseteq E$.
 Furthermore, by definition of $\bar{w}$, we have $\bar{w}(\bar{P}) = p(P)-w(P)+\sum_i p(C_i)-w(C_i) \ge p(P)-w(P)$, as
 we assume that $p(C) \ge w(C)$ holds for all cycles. Hence, indeed, $\bar{w}(\bar{P}) \ge 0$, as claimed.\qed
 \end{proof}

 \subsection{The Case $b=3$}

To complement Theorem~\ref{t-core} we will prove that the case $b\equiv 3$ is co-\NP-complete even for multiple partners assignment games. In order to do this we reduce from the
{\sc Cubic Subgraph} problem, which is that of testing whether a graph has a {\it 3-regular subgraph}, that is, a subgraph in which every vertex has degree~3.
Plesn\'ik~\cite{Pl84} proved that {\sc Cubic Subgraph} is \NP-complete even for planar bipartite graphs of maximum degree~4.

\begin{theorem}\label{t-core2}
The problem of deciding whether an allocation is in the core of a multiple partners assignment game is co-\NP-complete if $b\equiv 3$ and $w\equiv 1$.
\end{theorem}

\begin{proof}
The problem is readily seen to be in co-\NP: given a ``certificate'' coalition~$S$ we can efficiently verify that a given allocation $p$ is not in the core by checking that
$v(S)>p(S)$ holds; the latter can be done in polynomial time due to Lemma~\ref{lrt}.
 To prove co-\NP-hardness we
reduce from the {\sc Cubic Subgraph} problem restricted to bipartite graphs.
Recall that this problem is \NP-complete~\cite{Pl84}.

Let $G=(N,E)$ be a bipartite graph that is an instance of {\sc Cubic Subgraph}.
Let $n=|N|$.
For each vertex $u\in N$ we add a set of five new vertices $a_u$, $b_u$, $c_u$, $x_u$, $y_u$,
so in total we add $5n$ extra vertices to $G$. We also add every edge between a vertex in $\{a_u,b_u,c_u\}$ and a vertex in $\{u,x_u,y_u\}$,
so in total we add $9n$ extra edges to $G$.
We call the resulting graph $G^*=(N^*,E^*)$ (see also Figure~\ref{f-last}).
Note that each $6$-tuple $\{a_u,b_u,c_u,u,x_u,y_u\}$ induces a complete bipartite graph $K_{3,3}$ in $G^*$ (and that $u$ is a cut vertex of $G^*$ that separates
$\{a_u,b_u,c_u,x_u,y_u\}$ and $N\setminus \{u\}$).
Moreover, $G^*$ is a bipartite graph as desired.

\begin{figure}
\begin{center}
\scalebox{0.75}{\includegraphics{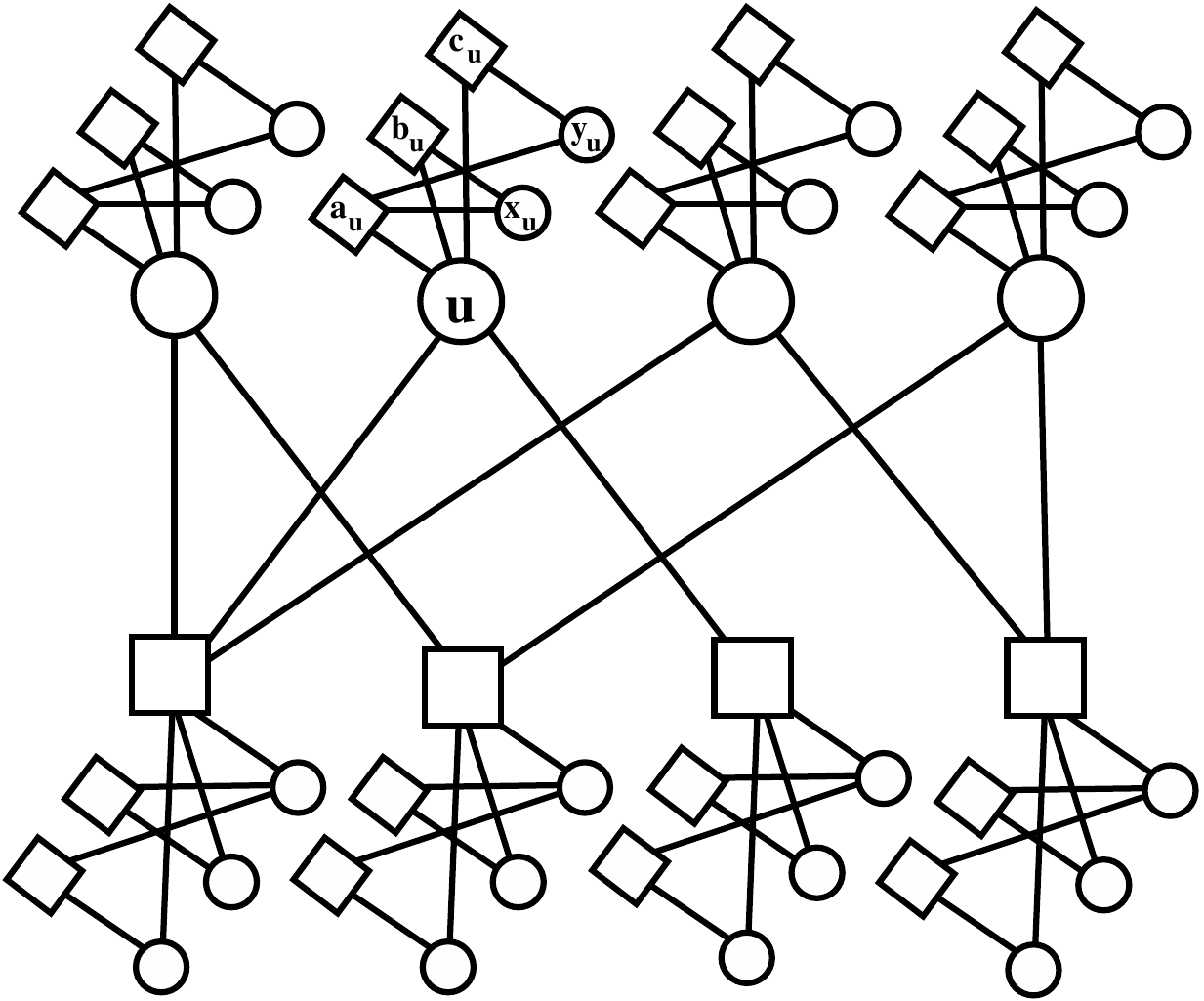}}
\caption{An example of the graph $G^*$ constructed in the proof of Theorem~\ref{t-core2}.}\label{f-last}
\end{center}
\end{figure}

Let $(N^*,v)$ be the multiple partners assignment game  defined on $(G^*,3,1)$, and
let $p$ be the vector, where we set, for every $u\in U$, $p(u)=\frac{3}{2}-\frac{1}{n}$ and $p(i)=\frac{3}{2}+\frac{1}{5n}$ for every $i\in \{a_u,b_u,c_u,x_u,y_u\}$. Then $p$ is an allocation,
as $v(N^*)=9n=p(N^*)$, where the first equality follows by taking a $b$-matching that consists of all nine edges in $G[\{a_u,b_u,c_u,u,x_u,y_u\}]$ (so, $9n$ edges in total).
We are left to prove that $G$ has a cubic subgraph if and only if $p$ does not belong to the core of $(N^*,v)$.

First suppose that $G$ has a 3-regular subgraph $H=(N_H,E_H)$. Then the coalition $N_H\subseteq N^*$ will not be satisfied: indeed, $v(N_H)=\frac{3}{2}|N_H|>\frac{3}{2}|N_H|-\frac{1}{n}|N_H|=p(N_H)$, so
$p$ does not belong to the core of $(N^*,v)$.

Now suppose that $G$ has no 3-regular subgraph. For contradiction, assume there exists a coalition $S\subseteq N^*$, such that
$p(S)<v(S)$. We choose $S$ such that $p(S)-v(S)$ is minimum, and moreover, we assume that $S$ is a coalition of smallest size over all such coalitions.

First consider the case that $S$ only contains vertices of $N$.
Then, as $G$ has no 3-regular subgraph, neither does $G[S]$, implying that $v(S)\leq \frac{3}{2}|S|-1$.  We use this inequality to find that
$p(S)=\frac{3}{2}|S|-\frac{1}{n}|S|\geq \frac{3}{2}|S|-1\geq v(S)$, a contradiction with our assumption that $p(S)<v(S)$. Hence $S$ contains at least
one vertex from a set $\{a_u,b_u,c_u,x_u,y_u\}$ for some $u\in N$. Let $T_u=S\cap \{a_u,b_u,c_u,x_u,y_u\}$, so we assume that $|T_u|\geq 1$.
If $|T_u|\leq 4$, then
$$p(S\setminus T_u)-v(S\setminus T_u)\leq p(S)-v(S)-|T_u|\left(\frac{3}{2}+\frac{1}{5n}\right) +|E(G[T_u\cup \{u\}])|\leq p(S)-v(S),$$
contradicting our choice of $S$.
Now suppose that $|T_u|=5$. If $u\notin S$, then
$$p(S\setminus T_u)-v(S\setminus T_u)\leq p(S)-v(S)-5\left(\frac{3}{2}+\frac{1}{5n}\right) +6\leq p(S)-v(S),$$
contradiction our choice of $S$ again. If $u\in S$, then
\[\begin{array}{lcl}
p(S\setminus (T_u\cup \{u\}))-v(S\setminus (T_u\cup \{u\})) &\leq &p(S)-v(S)-5\displaystyle\left(\frac{3}{2}+\frac{1}{5n}\right)-\left(\frac{3}{2}-\frac{1}{n}\right)+9\\[5pt]
&\leq &p(S)-v(S),
\end{array}\]
another contradiction with our choice of $S$.
Hence we conclude that a coalition $S$ with $p(S)<v(S)$ does not exist. In other words, $p$ belongs to the core of $(N^*,v)$.
This completes the proof of Theorem~\ref{t-core2}.
\qed
\end{proof}

\section{Future Work}\label{s-con}

We finish our paper with some directions for future research.
Our reduction from SFP to SRP
might be used to generalize more known results from SRP to SFP.
For instance, can we generalize the path to stability result of Bir\'o et al.~\cite{BBGKP13} for SRP to be valid for SFP as well?

\begin{table}[h]
\centering
\begin{tabular}{|c|c|c|c|}
\hline
     & P1 & P2     &P3\\
\hline
assignment games & yes &P &P\\
\hline
matching games  &P &P &P\\
\hline
mp assignment games  &yes &P &\hspace*{4.5mm}P if $b\leq 2$\\
				&           &   &NP-c if $b\equiv 3$\\		
\hline
mp matching games  &P if $b\leq 2$     & P if $b\leq 2$     & \hspace*{4.5mm}P if $b\leq 2$ \\
             			 &? if $b\equiv 3$ & ? if $b\equiv 3$ &NP-c if $b\equiv 3$\\
\hline			
\end{tabular}
\vspace*{5mm}
\caption{The four cooperative game models and their status with respect to problems
P1 (testing core non-emptiness), P2 (finding a core allocation) and P3 (verifying core membership). Here
``yes'' means that every instance of the problem under consideration is a yes-instance, whereas ``P'' smeans polynomial time and`` \NP-c'' means being \NP-complete. As before,``mp'' stands for ``multiple partners''.
Note that the rows for assignment games and matching games were known already (see Section~\ref{s-intro}) but can be generated from the rows for the other two models. Also the first two entries in the row for multiple partners assignment games were known already (Theorem~\ref{sot92}). Except the two open cases indicated by a
``?'', all other results follow from Theorems~\ref{t-core} and~\ref{t-core2}, which are proven in this paper.}\label{t-tabtab}
\end{table}

In contrast to multiple partners assignment games, we
do not know the complexity of testing core non-emptiness (problem~P1) and finding a core allocation (problem~P2) for multiple partners matching games $(G,b,w)$ when $b\equiv 3$.
We pose these two problems as open problems (see also Table~\ref{t-tabtab}).
As we showed in Theorem~\ref{t-core} we can decide in polynomial-time whether the core of a multiple partners matching is nonempty if $b\leq 2$. This immediately raises the question what to do if we find that the core is non-empty. One approach, which was recently considered for the case $b\equiv 1$,
is that of verifying whether the core can become nonempty after a small modification of the graph via adding or deleting some vertices or
edges~\cite{AHS16,BCKPS15,IKKKO16} or via a fractionally increase of the edge weights~\cite{CGKPSW}.

Theorem~\ref{t-nonemptyempty} states that there  exist multiple partners matching games with a non-empty core but for which the corresponding instance of SFP has no stable solutions.
It would be interesting to detect necessary and sufficient conditions on multiple partners matching games with a non-empty core,
such that the corresponding instance of SFP has a stable solution.
Another direction is to consider whether it is possible to obtain axiomatic characterizations of the core of multiple partners assignment games, or even for (multiple partners) matching games, similar to those that exist in the literature for assignment games (see, for instance,~\cite{Sa95,To05}).

Another major open problem is that of determining the complexity of computing the so-called nucleolus of a matching game $(N,v)$. So far only partial results
have been obtained in this direction, namely
polynomial-time algorithms for assignment games~\cite{SR01} and for matching games with unit edge weights~\cite{KP03}
or, more generally, with edge weights that can be expressed as the sum of positive weights on the incident vertices of the graph $G=(N,E)$~\cite{Pa01}, or even more generally, with arbitrary edge weights as long as a certain condition on the least core of a subgraph (with vertex weights) of $G$ is satisfied~\cite{Fa15}.
 Can these results be generalized for $b\leq 2$?

Aziz and de Keijzer~\cite{AK14} showed that the Shapley value of a matching game can be computed in polynomial time for
graphs of maximum degree~2, and for graphs that have a small modular decomposition into cliques or cocliques.
Recently, Bousquet~\cite{Bo} extended these results by
proving that the Shapley value of a matching game can be computed in polynomial time for trees.
Hence, another research direction would be to consider the Shapley value for multiple partners matching games restricted to trees,
or even to trees of small maximum degree.

Chalkiadakis et al.~\cite{CEMPJ10} defined cooperative games with overlapping coalitions, where players can be involved in coalitions with different intensities, leading to
three alternative core definitions.
It would be interesting to study
the problem of finding a stable solutions and
problems~P1--P3
in these settings.
To illustrate this, suppose that the set of soccer teams from the example
at the start of Section~\ref{s-intro} consists of international teams.
Then it seems realistic that the home team needs to spend fewer days for playing the game than
the visiting team, which must travel  from another country.
Hence,
every team now has a number of days for playing friendly games instead of an upper limit (target) on the number of such games.
As another example of two-player coalitions that require different intensities, consider a PhD student, who can be expected to spend significantly more time per week on her PhD project per than her supervisor.

\bibliographystyle{plain}

\end{document}